\documentclass{aims}
\usepackage{amsmath}
\usepackage{bm}
\usepackage[normalem]{ulem}

  \usepackage{paralist}
  \usepackage{graphics} 
  \usepackage[caption=false]{subfig}
  \usepackage{epsfig} 
\usepackage{graphicx}  \usepackage{epstopdf}
 \usepackage[colorlinks=true]{hyperref}

\hypersetup{urlcolor=blue, citecolor=red}

\usepackage[capitalize]{cleveref}
\def\currentvolume{X}
 \def\currentissue{X}
  \def\currentyear{200X}
   \def\currentmonth{XX}
    \def\ppages{X--XX}

\newtheorem{theorem}{Theorem}[section]

\newtheorem{proposition}{Proposition}

\theoremstyle{definition}
\newtheorem{definition}[theorem]{Definition}

\newcommand{\ep}{\varepsilon}

\newcommand{\pd}[2]{\frac{\partial #1}{\partial #2}}
\newcommand{\pdd}[2]{\frac{\partial^2 #1}{\partial #2^2}}
\newcommand{\pddm}[3]{\frac{\partial^2 #1}{\partial #2 \partial #3}}

\newcommand{\alk}[1]{{#1}}

\newcommand{\tr}{\mbox{\rm{tr}}}
\newcommand{\beq}{\begin{equation}\begin{aligned}}
\newcommand{\eeq}{\end{aligned}\end{equation}}

\newcommand{\uv}{\bm{u}}
\newcommand{\us}{\bm{u^*}}
\newcommand{\usp}{\bm{u^*}_x}
\newcommand{\w}{\bm{w}}
\newcommand{\q}{\bm{q}}
\newcommand{\ps}{\bm{p_*}^\pm}
\newcommand{\st}{(\bm{s}_{\bm *}^{\pm})^T}

\newcommand{\D}{\bm{D}}
\newcommand{\J}{\bm{J}}

\newcommand{\f}{\bm{f}}
\newcommand{\Bl}{\bm{B}_\lambda}
\newcommand{\Bz}{\bm{B_0}}

\newcommand \bk{\color{black}}

\newcommand \bit{\begin{itemize}}
\newcommand \eit{\end{itemize}}

\title[Heterogeneity Localizes Turing Patterns] 
      {Spatial Heterogeneity Localizes Turing Patterns in Reaction-Cross-Diffusion Systems}

\author[E. A. Gaffney, A. L. Krause, P. K. Maini, C. Wang]{}


\thanks{{Authors in alphabetical order.} $^*$ Corresponding author: Andrew.Krause@durham.ac.uk}

\begin{document}
\maketitle

\centerline{\scshape Eamonn A. Gaffney}
\medskip
{\footnotesize
 \centerline{Wolfson Centre for Mathematical Biology, Mathematical Institute,} \centerline{University of Oxford, Andrew Wiles Building, Radcliffe Observatory Quarter,} \centerline{Woodstock Road, Oxford, OX2 6GG, United Kingdom}
} 

\medskip

\centerline{\scshape Andrew L. Krause$^*$}
\medskip
{\footnotesize
 \centerline{  Mathematical Sciences Department, Durham University,}
 \centerline{Upper Mountjoy Campus, Stockton Rd,} 
 \centerline{Durham DH1 3LE, United Kingdom}
}

\medskip

\centerline{\scshape Philip K. Maini}
\medskip
{\footnotesize
 \centerline{Wolfson Centre for Mathematical Biology, Mathematical Institute,} \centerline{University of Oxford, Andrew Wiles Building, Radcliffe Observatory Quarter,} \centerline{Woodstock Road, Oxford, OX2 6GG, United Kingdom}
}

\medskip

\centerline{\scshape Chenyuan Wang}
\medskip
{\footnotesize
 \centerline{Wolfson Centre for Mathematical Biology, Mathematical Institute,} \centerline{University of Oxford, Andrew Wiles Building, Radcliffe Observatory Quarter,} \centerline{Woodstock Road, Oxford, OX2 6GG, United Kingdom}
} 

\medskip


\begin{abstract}
Motivated by bacterial chemotaxis and multi-species ecological interactions in heterogeneous environments, we study a general one-dimensional reaction-cross-diffusion system in the presence of spatial heterogeneity in both transport and reaction terms. Under a suitable asymptotic assumption that the transport is slow over the domain, while gradients in the reaction heterogeneity are not too sharp, we study the stability of a heterogeneous steady state approximated by the system in the absence of transport. Using a WKB ansatz, we find that this steady state can undergo a Turing-type instability in subsets of the domain, leading to the formation of localized patterns. The boundaries of the pattern-forming regions are given asymptotically by `local' Turing conditions corresponding to a spatially homogeneous analysis parameterized by the spatial variable. We developed a general open-source code which is freely available, and show numerical examples of this localized pattern formation in a Schnakenberg cross-diffusion system, a Keller-Segel chemotaxis model, and the Shigesada-Kawasaki-Teramoto model with heterogeneous parameters. We numerically show that the patterns may undergo secondary instabilities leading to spatiotemporal movement of spikes, though these remain approximately within the asymptotically predicted localized regions.  This theory can elegantly differentiate between spatial structure due to background heterogeneity, from spatial patterns emergent from Turing-type instabilities.
\end{abstract}

\section*{Dedication} We would like to dedicate this paper to the memory of Professor Masayasu (Mayan) Mimura, one of the pioneers of mathematical biology. Not only was Mayan an exceptional mathematician, he would also enthusiastically encourage
others and was particularly kind and generous to early career researchers. He always brought a smile to every face.

\section{Introduction}
A major extension to Turing's Chemical Theory of Morphogenesis \cite{turing1952chemical} is the incorporation of transport beyond Fickian diffusion such as models of chemotaxis \cite{keller1971model,horstmann20031970,hillen2009user} and general cross-diffusion \cite{mimura1980spatial, choi2004existence, le2005regularity, ruiz2013mathematical, gambino2016super, ritchie_hyperbolic_2020}. While the basic ideas of linear stability analysis employed by Turing to predict pattern formation in reaction-diffusion systems extends readily to such systems \cite{murray2004mathematical, krause_near_2021}, the inclusion of spatial heterogeneity  leads to difficulties in predicting pattern formation using linear analysis. Yet even Turing himself was well-aware that few biological systems actually form spontaneous spatially-structured patterns from purely homogeneous ones. Rather, most organisms evolve from some complex spatial state to another spatially patterned state during different stages of development. Such `hierarchical' pattern formation arises from interacting systems on different temporal and spatial scales is likely commonplace in embryonic development \cite{maini1995hierarchical, krause2021isolating}, with some authors suggesting that such mechanisms can be modelled via cross-diffusion systems \cite{roussel2004reaction}. In evolutionary and ecological settings, where cross-diffusion-type interactions are especially well-motivated \cite{okubo2001diffusion}, spatial heterogeneity is extremely important, accounting for innate landscape and demographic variation \cite{cantrell2004spatial}. In such settings it is especially important to be able to determine if an observed variation in population density is due to environmental heterogeneity, or due to species interactions. In this paper, we develop a theory to predict and understand pattern formation in spatially-heterogeneous cross-diffusion systems, under the assumption of a scale separation between the heterogeneity and the pattern wavelength. Within this asymptotic regime, our theory can distinguish between spatial structure due to background heterogeneity and spatial patterns due to Turing-type instabilities. 

Spatial heterogeneities in reaction-diffusion systems have been (numerically) shown to change local instability conditions for pattern formation \cite{benson1993diffusion, page2003pattern}, modulate the size and wavelength of patterns \cite{page2005complex}, and localize (or pin) spike patterns in space \cite{iron2001spike, ward2002dynamics, wei2017stable}, which can be studied analytically for certain systems in a semi-strong interaction regime. The presence of even simple spatial heterogeneity in reaction-diffusion systems can induce spatiotemporal effects, such as changing the stability of patterned states and leading to the movement of spike solutions \cite{krause2018heterogeneity, kolokolnikov2018pattern}. In regimes of highly localized spike solutions \cite{iron2001spike, ward2002dynamics, wei2017stable}, or highly localized heterogeneities \cite{doelman2018pulse}, some specific models are analytically tractable. There are some studies on bifurcation structures in spatially-heterogeneous cross-diffusion systems \cite{kuto2009bifurcation, cai2016fish} for specific models.  Nevertheless, there are few tools for understanding spatially-heterogeneous reaction-(cross)-diffusion systems with the same level of generality, and giving the same level of insight, as Turing's original use of linear stability analysis.

The typical Turing-type linear stability analysis proceeds by first linearizing the model about a spatially homogeneous steady state, and then exploiting eigenvalues of the spatial operators (the Laplacian in the context of general cross-diffusion). Using an eigenfunction (of the Laplacian) and exponential in time ansatz, one can reduce the study of instabilities to the computation of eigenvalues of a given matrix. These eigenvalues then indicate linear (in)stability of a particular eigenfunction, and hence one can develop an idea of what spatial perturbations might grow, and hence form spatially patterned states. In the limit of a sufficiently large domain, one can use the approximation of a continuous spatial spectrum to derive conditions involving only the parameters of the model which are necessary for such pattern-forming instabilities, and which become sufficient on large enough domains \cite{murray2004mathematical}. Such conditions are a valuable result from the linear analysis, as they often give biologically-interpretable insight into general classes of models, such as the celebrated short-range activation/long-range inhibition theory of pattern formation in two-species reaction-diffusion systems \cite{meinhardt2012turing}. 

This analysis, however, requires the existence of a spatially homogeneous equilibrium, and the spectral analysis of a scalar operator (the Laplacian). In general spatially-heterogeneous systems, such assumptions are violated. In addition to the difficulties inherent in computing steady states of these spatially-heterogeneous systems, linearization about such a state typically leads to linear systems involving spatial operators coupling multiple species. Theory for such non-scalar Sturm-Liouville problems does exist (see Section 3.1 of \cite{van2019diffusive} for discussion and references), but it is typically not useful for analytical computations as one has to determine the spectrum via numerical approximations anyway.  If a spatially-heterogeneous steady state can be found, one can use Galerkin expansions to compute instability conditions by truncating an infinite-matrix system (arising from coupled modes which are `diagonalized' in the homogeneous setting) \cite{kozak2019pattern, van2021pattern}. However, this analysis is quite involved for a given system, and does not often lead to general biological insights such as those gained by deriving necessary conditions for pattern formation in the homogeneous setting. From a modelling perspective there is also a difficulty in differentiating between spatial structure arising from nonlinear interactions and instability, and spatial structure due to the underlying heterogeneity. As mentioned, determining precisely what causes an observed spatial variation in a population or a developing organism is extremely valuable for mechanistic understanding, and eventually for influencing or controlling the system such as in conservation ecology or tissue engineering. 

Recently, in the context of spatial heterogeneity in the kinetics of one-dimensional reaction-diffusion systems, we proposed a theory of linear stability addressing the above challenges \cite{krause_WKB}. We assumed a separation of scales allowing us to write the model in a limit of asymptotically small transport relative to kinetic interactions. Such a limit in the spatially homogeneous case can be related to assuming a sufficiently large domain so that the spatial spectrum of the Laplacian can be well-approximated as a continuous variable. In this regime, we used a WKB ansatz to study the stability of a heterogeneous steady state, finding that the usual Turing instability conditions could be satisfied locally as if they were simply parameterized by the spatial variable. Such a localization of these conditions was shown to precisely correspond to where full numerical simulations found patterned states deviating from a heterogeneous steady state. Another key insight was that this steady state could be well-approximated by zeros of the reaction kinetics, and hence became analytically amenable in a wide class of systems.

In this paper, we generalize our approach to study two-species reaction-cross-diffusion systems which may include nonlinearity and heterogeneity in all transport and reaction terms. We formulate our model and linearize about a steady state in \cref{Model_Sect}. We compute general conditions for pattern formation in the spatially homogeneous version of the model in \cref{Hom_Cond_Sect}, and state corresponding conditions in the spatially heterogeneous case in \cref{Het_Cond_Sect}. In \cref{WKB_Sect} we solve the heterogeneous linear problem, and discuss the problem of mode selection from this solution. Using this idea, in \cref{prop2sect} we derive the conditions for pattern formation in the heterogeneous case. We give a variety of numerical examples in \cref{Numerical_Sect}, illustrating the utility of our theory even in cases where emergent patterns become spatiotemporally complex. Finally we discuss a range of open mathematical problems and possible applications in \cref{Discuss_Sect}.

\section{Model Formulation \& Conditions for Pattern Formation}\label{Model_Sect}
As we are motivated by finding simple conditions for pattern formation, we will focus on an asymptotic regime where transport is assumed small. In \cite{krause_WKB}, we related this asymptotic scale, given by $\ep$, to diffusion coefficients, length and time scales of relevance in developmental settings, but note that such an asymptotic regime is always needed to deduce algebraic conditions for Turing conditions which are independent of the spectrum of the spatial operator; see Chapter 2 of \cite{murray2004mathematical} for further discussion of this `large domain' approximation. 

We consider the nonlinear heterogeneous cross-diffusion system, 
\beq
    \pd{u}{t} = \ep^2 \pd{}{x}\left (D_{11}(u,v,x)\pd{u}{x}+D_{12}(u,v,x)\pd{v}{x} \right ) + f(u,v,x),
\eeq
\beq
    \pd{v}{t} = \ep^2 \pd{}{x}\left (D_{21}(u,v,x)\pd{u}{x}+D_{22}(u,v,x)\pd{v}{x} \right ) + g(u,v,x),
\eeq
where we have nondimensionalized the model to be in $x \in [0,1]$, and assume sufficient regularity on the six nonlinear functions $D_{ij}, f, g$. We also henceforth assume $0 < \ep \ll 1$, noting that  in terms of a physical diffusion scale, $D$, a lengthscale $L$ and timescale $T$, we have from model non-dimensionalization that 
\begin{equation}\label{eps}
{\ep^2} = \frac {DT}{L^2}, 
\end{equation}
which continuously decreases as the scale of diffusion decreases {for instance.}

More compactly we write this system as
\beq\label{orig_eqn}
    \pd{\uv}{t} = \ep ^2 \pd{}{x}\left ( \D(\uv,x)\pd{\uv}{x}\right) + \f(\uv,x), 
\eeq
with
\beq
    \uv = \begin{pmatrix}u\\ v \end{pmatrix}, \,\, \D(\uv,x) = \begin{pmatrix} D_{11}(u,v,x) & D_{12}(u,v,x) \\
    D_{21}(u,v,x) & D_{22}(u,v,x) \end{pmatrix}, \,\, \f(\uv,x) = \begin{pmatrix}f(u,v,x)\\ g(u,v,x) \end{pmatrix}.
\eeq
We assume for all $\uv \in \mathbb{R}^2$ (or in a suitably chosen subset) and $x \in [0,1]$ that $\D(\uv,x)$ is positive-definite \alk{(that is, all of its eigenvalues remain positive definite)}. This implies that no-flux and Neumann boundary conditions are equivalent, so for concreteness we write
\beq\label{BCs}
    \pd{u}{x}(t,0) = \pd{u}{x}(t,\alk{1}) = \pd{v}{x}(t,0) = \pd{v}{x}(t,\alk{1}) = 0, \,\, \textrm{for all } t\geq 0.
\eeq

We let $\alk{\widehat{\us}} = (u_s, v_s)^T$ be a steady state of our system, so that it satisfies
\beq\label{steady_state}
    \bm{0} = \ep ^2 \pd{}{x}\left ( \D(\alk{\widehat{\us}},x)\pd{\alk{\widehat{\us}}}{x}\right) + \f(\alk{\widehat{\us}},x),
\eeq
as well as the Neumann boundary conditions \eqref{BCs}. In general finding such a steady state analytically is extremely hard, but we can approximate it in the limit of small $\ep$. If $\ep=0$, we have that a heterogeneous steady state would satisfy \begin{equation}\label{approx_steady_state}
\f(\us,x)=\bm{0},
\end{equation}
which coincides exactly with a homogeneous steady state in the case that $\us$ does not depend on $x$. If we assume that such a \alk{function} $\us$ has sufficiently well-behaved derivatives, that is $|\partial \us/\partial x| =  o(1/\ep)$,
and satisfies the boundary conditions \eqref{BCs}, then it is an asymptotic approximation to a solution of \cref{steady_state}\alk{, that is $\us = \widehat{\us} +O(\ep^2)$. In the following we will not distinguish between $\widehat{\us}$ satisfying \cref{steady_state}, and $\us$ satisfying \cref{approx_steady_state}, as we will only expand to order $\ep$.}

We now consider linear stability of a steady state satisfying \eqref{steady_state}, introducing another small parameter $|\alk{\delta}| \ll 1$ which will be asymptotically smaller than $\ep$. We expand our solutions as $\uv = \us(x) + \alk{\delta} \w(t,x)$ and substitute this into \eqref{orig_eqn} to find
\beq
    \alk{\delta}\pd{\w}{t} = \ep ^2 \pd{}{x}\left ( \D(\us+\alk{\delta} \w,x)\left( \pd{\us}{x}+\alk{\delta} \pd{\w}{x}\right)\right) + \f(\us,x)+\alk{\delta}\J(\us,x)\w + O(\alk{\delta}^2),
\eeq
where we have expanded the kinetics $\f$ in a Taylor series about $\alk{\delta}$, and hence $\J$ is the Jacobian matrix evaluated at the approximate steady state. Writing $\usp$ as the derivative of $\us$ with respect to $x$, we expand the transport term to find
\beq
\begin{aligned}
     \pd{}{x}&\left ( \D(\us+\alk{\delta} \w,x) \left (\usp+ \alk{\delta} \pd{\w}{x}\right)\right) \\ =  \pd{}{x}&\left ( \D(\us,x) \usp \right)  + \alk{\delta}\left (\D(\us,x) \pdd{\w}{x}+ \bm{M}\pd{\w}{x}+\bm{N}\w  \right)+ O(\alk{\delta}^2), 
\end{aligned}
\eeq
where 
\beq
    \bm{M} = \begin{pmatrix} \pd{D_{11}}{x}+\pd{D_{11}}{u}u_x^*+\pd{D_{12}}{u}v_x^* & \pd{D_{12}}{x}+\pd{D_{11}}{v}u_x^*+\pd{D_{12}}{v}v_x^* \\
    \pd{D_{21}}{x}+\pd{D_{21}}{u}u_x^*+\pd{D_{22}}{u}v_x^* & \pd{D_{22}}{x}+\pd{D_{21}}{v}u_x^*+\pd{D_{22}}{v}v_x^* \end{pmatrix},
\eeq
and
\beq
    \bm{N} = \begin{pmatrix} \pd{D_{11}}{u}u_{xx}^*+\pd{D_{12}}{u}v_{xx}^*+\pddm{D_{11}}{u}{x}u_x^*+\pddm{D_{12}}{u}{x}v_x^* & \pd{D_{11}}{v}u_{xx}^*+\pd{D_{12}}{v}v_{xx}^*+\pddm{D_{11}}{v}{x}u_x^*+\pddm{D_{12}}{v}{x}v_x^* \\
    \pd{D_{21}}{u}u_{xx}^*+\pd{D_{22}}{u}v_{xx}^*+\pddm{D_{21}}{u}{x}u_x^*+\pddm{D_{22}}{u}{x}v_x^* & \pd{D_{21}}{v}u_{xx}^*+\pd{D_{22}}{v}v_{xx}^*+\pddm{D_{21}}{v}{x}u_x^*+\pddm{D_{22}}{v}{x}v_x^* \end{pmatrix},
\eeq
and we have suppressed the dependence of $D_{ij} = D_{ij}(u^*,v^*,x)$. {Throughout this paper, we will assume that $\D,~ \D^{-1}, ~\J$ have smooth and bounded coefficients, in particular with smooth and bounded derivatives with respect to $x$, and also with respect to $u$ and $v$. Thus such properties will be inherited by $\bm{M}$ and $\bm{N}$.} \alk{We also remark that, as $\us$ does not depend on $\ep$, all of these matrices are independent of $\ep$.}

Discarding the {ord$(1)$} terms\footnote{\alk{We use the notation ord$(\ep^k)$ as shorthand for ``asymptotically of order $k$ in $\ep$."}} by using \eqref{steady_state}, and neglecting terms of $O(\alk{\delta}^2)$, we obtain the linear system,
\beq\label{linear_RDS}
    \pd{\w}{t} = \ep ^2  \left (\D(x) \pdd{\w}{x}+ \bm{M}(x)\pd{\w}{x}+\bm{N}(x)\w  \right) + \J(x)\w,
\eeq
where we have omitted the explicit dependence on $\us$ and its derivatives. From now on we will drop dependence on the steady state $\us$ and view these four matrices $\D, \J, \bm{M}$ and $\bm{N}$ as simply depending on $x$. 
While this system appears much more complicated than that studied in \cite{krause_WKB}, in fact we will show that the impact of  $\bm{M}$ will be confined to the structure of the unstable modes, and that both $\bm{M}$ and $\bm{N}$ will not have an influence on the conditions for instability and hence pattern formation. While the diffusion tensor $\D$ is more general, not being confined to a diagonal matrix, its positive-definiteness will be sufficient to carry out a procedure and derive results analogous to that in \cite{krause_WKB}, which generalize conditions for pattern-forming instabilities to occur. Before discussing these in detail, we first review the spatially homogeneous results on such instabilities from a slightly different perspective.

\subsection{Spatially Homogeneous Instability Criteria}\label{Hom_Cond_Sect}
If we assume that $\D$ and $\f$ do not depend on $x$, then \eqref{orig_eqn} admits a constant steady state $\us$. Hence, the linearized system \eqref{linear_RDS} without spatial dependencies can be written as
\beq\label{linear_RDS_homog}
        \pd{\w}{t} = \ep ^2 \D\pdd{\w}{x}+ \J\w,
\eeq
where the matrices $\D$ and $\J(\us)$ are constant. We consider the usual expansion into eigenmodes of the Laplacian, $\w \propto e^{\lambda t}\cos(k x)$, where $k$ takes discrete values to satisfy the boundary conditions, and hence we require $k/\pi$ to be an integer for Neumann conditions. We then have that $\lambda$ is an eigenvalue of $\J - (\ep k)^2\D$.

The usual way to derive instability conditions is to use the polynomial dispersion relation given by,
\beq\label{hom_disp_polyn}
    \lambda^2 - \tr(\J - (\ep k)^2\D)\lambda + \det(\J - (\ep k)^2\D) = 0,
\eeq
and then require $\Re(\lambda)<0$ for $k=0$, and $\Re(\lambda)>0$ for some $k>0$. The first of these entails that
\beq\label{homog_instab_cond_1}
    \tr(\J) < 0, \quad \det(\J)>0,
\eeq
for stability of $\us$ in the absence of diffusion. We also have that, since $\D$ is positive definite it has $\tr(\D)>0$ and hence by linearity of the trace we have $\tr(\J - (\ep k)^2\D) = \tr(\J) - (\ep k)^2\tr(\D) < \tr(\J) <0$. For an instability for $k>0$ we then must have a positive growth rate which entails,  for $\lambda$ real, 
\beq \label{lam1} 
    2\lambda = \tr(\J - (\ep k)^2\D) + \sqrt{[\tr(\J - (\ep k)^2\D)]^2-4\det(\J - (\ep k)^2\D))}>0.
\eeq
If $\lambda$ is not real, then the above implies that $\Re(\lambda)<0$, and hence we only consider real growth rates, unless explicitly stated otherwise. The usual approach is to then maximize this growth rate as a function of $k$, and require $k>0$. Instead, we will consider conditions in terms of {permitted} values of $k$ to begin with, as this will generalize to the heterogeneous case as in \cite{krause_WKB}. 

For the marginal stability curve given by $\lambda=0$, we  see that the only remaining term of \eqref{hom_disp_polyn} implies that $(\ep k)^2$ is an eigenvalue of $\Bz = \D^{-1}\J$. Solving for the eigenvalues of this matrix,   we then have that
\beq
    2(\ep k)^2 = \tr(\Bz) \pm \sqrt{[\tr(\Bz)]^2-4\det(\Bz)}>0,
\eeq
are the two places where \alk{the graph of $\Re($}$\lambda(k^2)\alk{)}$ given by \eqref{hom_disp_polyn} crosses the  $k^2$-axis. For there to be a positive and real range of  $k^2$, we need 
\beq \label{ce1} 
    \tr(\Bz) + \sqrt{[\tr(\Bz)]^2-4\det(\Bz)}>0. 
    \eeq
   We first consider $\Re (\lambda)$ as this is always  defined, real  and continuous as $k^2$ varies.  In particular for $k^2$ sufficiently large
   we have that the expression for $\lambda$ in \cref{lam1} gives $\Re (\lambda) <0 .$ 
Thus if, in addition to \cref{ce1}, we have 
    \beq
    \tr(\Bz) - \sqrt{[\tr(\Bz)]^2-4\det(\Bz)} < 0, 
    \eeq
   then  \alk{the graph of }$\Re (\lambda(k^2))$ crosses the $k^2\geq 0$ axis once only, 
   with   $\Re (\lambda) <0$ for large $k^2$ and  $\Re (\lambda)$ defined, real and continuous: this is sufficient to imply   $\Re (\lambda)>0$ as $k^2\rightarrow 0$.
 In turn, this  
    contradicts the requirement of stability   at $k=0$. 
    Hence,   for an instability range of $k^2$, where $\Re (\lambda ) >0$,   which does not contradict the stability requirement at $k=0$, we require 
     \beq
    \tr(\Bz) - \sqrt{[\tr(\Bz)]^2-4\det(\Bz)} > 0. 
    \eeq
     Noting that $\det(\J)>0 \implies \det(\Bz)>0$ as $\D^{-1}$ is positive-definite, this then demands  the conditions
    \beq\label{homog_instab_cond_2}
    \tr(\Bz) > 0 \textrm{ and } [\tr(\Bz)]^2-4\det(\Bz)>0.
\eeq
Conditions \eqref{homog_instab_cond_1} and \eqref{homog_instab_cond_2} are precisely the usual necessary conditions for Turing instability, which become sufficient in the limit of $\ep \to 0$. We summarize these as:

\begin{proposition}\label{hom_prop}
Let $0 < \ep \ll 1$ and  assume $\J$ and $\D$ are constant matrices for all $x \in [0,1]$. If we assume stability to homogeneous perturbations, i.e.~the inequalities \eqref{homog_instab_cond_1} are satisfied
then, subject to a wave selection constraint, there exists a non-homogeneous perturbation $\bm w$ satisfying \cref{linear_RDS} and homogeneous Neumann conditions at $x\in\{0,1\}$  that  grows exponentially in time in the interval $x \in [0,1]$ if the inequalities \eqref{homog_instab_cond_2} are satisfied.
\end{proposition}
These are precisely the usual Turing instability conditions for the case of cross-diffusion (that is, when  positive definite $\D$ can be a full matrix). See \cite{ritchie_hyperbolic_2020} for general examples of such homogeneous cross-diffusion systems, with identical instability criteria given there by equations (3.15)-(3.16).  Furthermore note that the wave selection constraint requires that the domain (or diffusion scale) is of a suitable size to fit a half-integer number of modes onto the domain, as required to satisfy the Neumann boundary conditions. This can always be achieved by continuously reducing the diffusion scale, $D$ (or equivalently $\ep$), for example,  and is detailed for example in Murray's text \cite{murray2004mathematical}.

\subsection{Spatially Heterogeneous Instability Criteria}\label{Het_Cond_Sect}

In the heterogeneous case, we require stability to  homogeneous perturbations across the whole domain to prevent such a mode destabilizing the system. {In particular,} we have the following  analogous heterogeneous result:
\begin{theorem}\label{het_prop}
Let $0 < \ep \ll 1$ and  assume that  $[\tr(\Bz(x))]^2-4\det(\Bz(x))$ has only simple zeros for all $x \in [0,1]$, where $\Bz(x) = \D^{-1}(x)\J(x)$.  We assume   stability to local homogeneous perturbations, i.e.
\beq \label{hcds} 
    \tr(\J(x)) < 0, \quad \det(\J(x))>0, \quad \textrm{for all } x \in [0,1]. 
\eeq
Then, subject to a wave selection constraint, {which can always be satisfied for a sufficiently small diffusion scale,} there 
 exists a non-homogeneous, bounded and non-trivial perturbation solution $\bm w$ satisfying \cref{linear_RDS} and homogeneous Neumann conditions at $x\in\{0,1\}$ that grows exponentially in time only within the interval $x \in \mathcal{T}_0$ if
\beq\label{het_instab_cond}
    \tr(\Bz(x)) > 0, \quad [\tr(\Bz(x))]^2-4\det(\Bz(x))>0, \quad \textrm{for all } x \in \mathcal{T}_0,
\eeq
where $\mathcal{T}_0$ is the largest subset of $[0,1]$ for which the conditions \eqref{het_instab_cond} hold.
\end{theorem}
While this statement is essentially identical to \alk{Instability Criterion} 2.2 in \cite{krause_WKB}, we remark that here $\D$ is a full matrix and it may  depend on $x$ explicitly or via $\us(x)$.  We then have that \cref{het_prop} provides  a local variant  for the inhomogeneous mode instability condition  of \cref{hom_prop}, albeit  restricted in a sense to the set $x \in \mathcal{T}_0$. 

\section{Asymptotic Solutions of the Linearized System}\label{WKB_Sect}

 We proceed to develop  leading $\ep$-order WKB solutions of \cref{linear_RDS} that will be used to deduce  \cref{het_prop}. In \cref{wkb_sec} we first determine the leading order general WKB solution in the regime where $\ep \ll 1$ and consider how to satisfy the homogeneous Neumann boundary conditions in \cref{admissible_growth_sec}, which leads to the wave selection constraint.
We then use the structure of this solution to determine the instability criteria and ultimately a deduction of \cref{het_prop} in \cref{prop2sect}.

\subsection{WKB Asymptotics}\label{wkb_sec}
We will now find an approximate solution to \cref{linear_RDS} for small $\ep$. As this system is linear, we can consider a solution which is separable in space and time of the form $\w = e^{\lambda t} \q(x)$. Writing $\q'$ as the derivative of $\q$ with respect to $x$, we find the problem for $\q$:
\beq\label{q_eqn}
        \bm{0} = \ep ^2  \left (\D(x) \q''+ \bm{M}(x)\q'+\bm{N}(x)\q  \right) + (\J(x) - \lambda \bm{I})\q.
\eeq
We now expand $\q$ using a WKB ansatz \cite{griffiths2018introduction, bender2013advanced} in the limit of small $\ep$ as  
\beq
    \q = \exp\left(\frac{\mathrm{i}\phi(x)}\ep\right)\bm{p}(x), \quad \bm{p}(x) = \bm{p_0}(x) + \ep \bm{p_1}(x) + O(\ep^2).
\eeq
Dropping the $x$ dependence for notational simplicity, we compute derivatives as
\beq
    \q' = \exp\left(\frac{\mathrm{i}\phi}\ep\right)\left(\bm{p}'+\frac{\mathrm{i}\phi'}{\ep}\bm{p}\right) = \exp \left [\frac{\mathrm{i}\phi}{\ep} \right ]\frac{\mathrm{i}\phi'}{\ep}\bm{p_0}+O(1),
\eeq
and
\beq
\begin{aligned}
    \q'' =& \exp\left(\frac{\mathrm{i}\phi}\ep\right)\left(
-\frac{\phi'^2}{\ep^2} \bm{p} + 
\frac 1 \ep\left(2\mathrm{i}\phi'\bm{p}'+ \mathrm{i}\phi'' \bm{p}\right) 
+ \bm{p}''
\right)\\
=& \exp\left(\frac{\mathrm{i}\phi}\ep\right)\left( 
-\frac{\phi'^2}{\ep^2}\bm{p_0}  + 
\frac{1}{\ep}\left({ -\phi'^2 \bm{p_1}} +2\mathrm{i}\phi'\bm{p_0}' + \mathrm{i}\phi''\bm{p_0} \right) 
\right)+O(1).
\end{aligned}
\eeq
We then have that the $O(1)$ approximation of \cref{q_eqn} is
\beq\label{ord1}
    \bm{0} = -\phi'^2\D\bm{p_0}+(\J - \lambda \bm{I})\bm{p_0} = \D[-\phi'^2\bm{I}+\Bl]\bm{p_0},
\eeq
where $$ \Bl := \D^{-1}(\J - \lambda \bm{I})$$ and $\bm{I}$ is the identity matrix. From this we see that $\phi'^2$ is an eigenvalue of $\Bl$,  of which there are two  (or one degenerate eigenvalue of algebraic multiplicity two), and we use $\mu^\pm_\lambda=\phi_{\pm}^{\prime 2}$ to label the  eigenvalues, which may conceivably be equal. Hence  we can solve \cref{ord1} by setting $\bm{p_0}(x) = Q^\pm_0(x)\bm{p_*}^\pm$ where $Q^\pm_0$ is a scalar function and $\bm{p_*}^\pm$ is the unit eigenvector associated with the eigenvalue $\mu^\pm_\lambda$. Hence we can compute $\phi$ as the solution of
\beq\label{phi_eqn}
\phi'^2_\pm =\mu_\lambda^\pm(x) \implies \phi_\pm = C^\pm_\phi + \int_{a_\pm}^x \sqrt{\mu_\lambda^\pm(\bar{x})}d\bar{x},
\eeq
where $\mu_\lambda^\pm(x)$ denote the two  eigenvalues of $\Bl$, with  $C^\pm_\phi$ and $a_\pm$   constants that are to be determined.   In order to compute $Q^\pm_0$, however, we must go to the next order in $\ep$.

The order $O(\ep)$ equation is given by
\beq
\begin{aligned}
    \bm{0} =&  \D({ -\phi'^2_\pm\bm{p_1}}+2\mathrm{i}\phi'_\pm\bm{p_0}' + \mathrm{i}\phi''_\pm\bm{p_0})+\mathrm{i}\phi'_\pm\bm{M}\bm{p_0}+(\J - \lambda \bm{I})\bm{p_1}\\
    =& \D[(-\phi'^2_\pm\bm{I}+\Bl)\bm{p_1} + 2\mathrm{i}\phi'_\pm\bm{p_0}' + \mathrm{i}\phi''_\pm\bm{p_0}+\mathrm{i}\phi'_\pm\D^{-1}\bm{M}\bm{p_0}],
\end{aligned}
\eeq
which implies that
\beq\label{ord2}
-[-\phi^{\prime 2}_\pm\bm{I}+\Bl]\bm{p_1}  
= \mathrm{i}[2\phi '_\pm Q^{\pm\prime}_0
+\phi''_\pm 
Q^\pm_0+\phi'_\pm 
Q^\pm_0\D^{-1}\bm{M}]
\ps
+2\mathrm{i}\phi'_\pm Q^\pm_0(\ps)'.
\eeq
The matrix premultiplying $\bm{p_1}$ has a zero {eigenvalue} by \cref{ord1}, and hence by the Fredholm Alternative Theorem we can compute a solvability condition to find $Q^\pm_0$ and $\ps$. Let $\st$ be the left eigenvector with zero eigenvalue and unit magnitude of $[(-\phi'^2_\pm\bm{I}+\Bl)]$. We multiply \cref{ord2} on the left by $\st$ to find the scalar equation
\beq
(2\phi'_\pm Q_0^{\pm\prime}+\phi''_\pm Q^\pm_0)\st\ps + \phi'_\pm 
 \alk{Q_0^{\pm}}\st\D^{-1}\bm{M}
\ps
+{ 2\phi'_\pm Q^\pm_0} \st(\ps)'  { =0.} \, \, \, \, \, \, 
\eeq
Solving this for $Q^\pm_0$ we find
\beq
\frac{Q_0^{\pm\prime}}{Q^\pm_0} = -\frac{\phi''_\pm }{2\phi'_\pm} - \frac{\st\D^{-1}\bm{M}\ps+2\st(\ps)'}{2\st\ps},
\eeq
which implies
\beq \label{q0eqn}
Q^\pm_0(x) = 
\frac{C^\pm_Q}{\sqrt{\phi'_\pm}}
\exp\left(-\int_{b_\pm}^x\frac{\st\D^{-1}\bm{M}\ps+2\st(\ps)'}
{2\st
\ps}
d\bar{x}\right),
\eeq where $C^\pm_Q$ are constants (possibly complex) and $b_\pm$ are  real constants. Without loss of generality, we set $b_\pm=a_\pm$ below by redefining $C^\pm_Q$. 

Finally, by {recasting} the constants $C_\phi$ and $C_Q$ appropriately, we can write our leading-order solution for $\w$ in the trigonometric form
\beq\label{w_sol}
\begin{aligned}
\w_\pm =& \frac{e^{\lambda t}}{(\mu_\lambda^\pm(x))^{\frac{1}{4}}}\exp\left(-\int_{a_\pm}^x\frac{\bm{s}_{\bm *}^\pm(\bar{x})^T\D^{-1}(\bar{x})\bm{M}(\bar{x})\ps(\bar{x})+2\bm{s}_{\bm *}^\pm(\bar{x})^T\ps(\bar{x})'}{2\bm{s}_{\bm *}^\pm(\bar{x})^T\ps(\bar{x})}d\bar{x}\right)\\
\times &\left[ C_0^\pm\cos\left(\frac{1}{\ep}\int_{a_\pm}^x \sqrt{\mu_\lambda^\pm(\bar{x})}d\bar{x} \right)+S_0^\pm\sin\left(\frac{1}{\ep}\int_{a_\pm}^x \sqrt{\mu_\lambda^\pm(\bar{x})}d\bar{x} \right) \right]\ps(x), 
\end{aligned}
\eeq
{where $C_0^\pm,~S_0^\pm$ are real constants. The constants $a_\pm$ are not independent degrees of freedom as shifts in these constants can be accommodated by changes in  $C_0^\pm,~S_0^\pm$ but it will be convenient to keep the above form for $\w_\pm$.} Furthermore, for  fixed $\lambda$, the solution $\w_\pm$ given by \cref{w_sol} gives two modes which satisfy \cref{linear_RDS} at leading order in $\ep$, corresponding to the two eigenvalues of $\Bl$ given by $\mu_\lambda^+$ and $\mu_\lambda^-$. We remark that this solution is nearly identical to the leading-order WKB solution in \cite{krause_WKB}, {except for} two notable differences. Firstly, there is an additional term involving the matrix $\bm{M}$ here, though this will not influence the conditions we find to ensure we have an instability, i.e.~a mode with $\Re(\lambda)>0$, given in \cref{het_prop}. {Secondly, the matrix $\D$ is no longer diagonal, though it is still positive-definite.} One technical point concerns the impact of  $\bm{M}$ and {especially $\D$ on the structure of $\w_\pm$ near singularities,  as will be discussed below and which is ultimately different} from the analysis presented in \cite{krause_WKB}.

\subsection{Admissible Growth Rates}\label{admissible_growth_sec}


{The} special case of $\D = \mathrm{diag}(1,d), \bm M=\bm N=\bm 0$ has previously been considered in
\cite{krause_WKB}, with a proof {of the analogue of Theorem \ref{het_prop} that is generalised to cross-diffusion below, together with a highlighting of  where the differences are. As in \cite{krause_WKB}, our} starting point is \cref{w_sol} and when such WKB solutions are associated with an instability that drives the system away from its steady state. That is, we look for WKB modes given by \cref{w_sol} which approximate a solution of \cref{linear_RDS} with $\Re(\lambda)>0$.

Our first consideration, directly analogous to \cite{krause_WKB}, is whether the  WKB solution is defined for all  $x \in[0,1]$. If so, leading order Neumann boundary conditions at $\{0,1\}$  entail 
\begin{equation}
  \label{eq:FundConstr00001}
    \int_0^1 \sqrt{\mu_{\lambda}^\pm(\bar{x})}\mathrm{d}\bar{x} = \frac 1 2  n\pi\ep, 
\end{equation} 
where, without loss of generality, $n$ is a positive integer and $a_\pm=S_0^\pm=0$, yielding a WKB cosine solution. Further  {given a suitable choice of a} sufficiently small $\ep$, the constraint given by \cref{eq:FundConstr00001} can be ensured simply by imposing 
\begin{equation}
  \label{eq:FundConstr01}
    \int_0^1 \sqrt{\mu_{\lambda}^\pm(\bar{x})}\mathrm{d}\bar{x}\in \mathbb{R}^+, 
\end{equation} 
where $\mathbb{R}^+$ denotes the positive reals and thus the square root in the integrand is the positive square root. The distinction between the two cases, in particular where the more general case \cref{eq:FundConstr01} holds, but \cref{eq:FundConstr00001} does not, is equivalent to the wave selection constraint in the standard Turing instability, which is a well-understood constraint that is  additional to the canonical Turing instability conditions of Eqs.~\eqref{homog_instab_cond_1},~ \eqref{homog_instab_cond_2} \cite{murray2004mathematical}.

However,  a non-trivial WKB solution \eqref{w_sol} need not be defined everywhere on the domain $x \in[0,1]$. More generally,  the steady state may be destabilized by  WKB solutions that are only non-zero on one or more intervals of the form $(a,b)$ with $0 \leq a < b \leq 1$. As will be investigated below, this will occur due to terms contributing to $\bm{w_\pm}$ in \cref{w_sol} becoming unbounded on approaching an interior point, say $x=a$. We therefore introduce `internal' boundary conditions at such points which allow us to match a non-zero WKB mode on an interval $(a,b)$, which we match to a zero solution outside of this interval. Thus, as demonstrated below, an internal homogeneous Dirichlet boundary condition $\bm{w_\pm}(a)=\bm{0}$ with $a_\pm=a$  in \cref{w_sol} is then necessary and sufficient for a well-defined solution local to $x=a$.  We may accommodate both cases -- homogeneous Neumann or internal Dirichlet  boundary conditions  -- by requiring
\begin{equation}
  \label{eq:FundConstr0001}
    \int_{a}^{b} \sqrt{\mu_{\lambda}^\pm(\bar{x})}\mathrm{d}\bar{x}\in \mathbb{R}^+
\end{equation} 
for a positive square root and $0 \leq a  < b  \leq 1$, so that the homogeneous boundary condition at $x=a, ~b$, whether it be Neumann or Dirichlet  , can be satisfied. In particular, this requires (i) an appropriate choice of the cosine or sine solution, according to whether respectively a Neumann or Dirichlet  condition is required at $x=a=a_\pm$ for the WKB solution of \cref{w_sol} and (ii) \cref{eq:FundConstr0001}, which allows the enforcement of the boundary condition at $x=b$ for suitably small $\ep$, with the constraint of $\ep$ constituting the wave selection constraint. 

Thus for a given non-trivial interval $(a,b)$, \cref{eq:FundConstr0001} constitutes the fundamental condition for an unstable WKB mode to destabilize the steady sate, assuming there is stability to homogeneous perturbations. Consequently, our aim below is to use \cref{eq:FundConstr0001} to deduce {\cref{het_prop},} including a characterization of the location of the instability regions, that is the intervals of the form $(a,b)$ in the above discussion, whose union constitutes the set ${\mathcal T}_0$ in the statement of {\cref{het_prop}.}

\section{Derivation of \cref{het_prop}}\label{prop2sect}

Extensive elements of the  reasoning  below are straightforward generalizations  of \cite{krause_WKB}, once expressions that involve the components of $\D$ are rewritten in terms of $\tr(\D)$ and $\det(\D)$ which are both positive since $\D$ is positive definite. However, a new technical result (given in \cref{p8.55}) is now required to handle a degenerate case arising at singular points, which can be very quickly dismissed for the simpler systems of \cite{krause_WKB} in contrast to the more general system considered here. In turn this     {alters the details of some of the derivations (given in \cref{nonorth,p9,newp}).}  We also present the propositions in a different order to \cite{krause_WKB}, with some alternative approaches, in a pedagogic attempt to simplify the derivation to aid understanding.

\subsection{Relating the Constraint to Conditions on ${\bf D},~{\bf J}$ and $\lambda$.}  

As in \cite{krause_WKB}, we  first define \emph{permissible} growth rates and eigenvalues which satisfy \eqref{eq:FundConstr0001}.

\begin{definition} A permissible pair $(\lambda, \mu^\pm_\lambda(x))$ is a tuple such that the value of $\lambda$ entails $\mu^\pm_\lambda(x)$ satisfies constraint \eqref{eq:FundConstr0001} for all $x$ in some non-empty interval $(a,b)\subseteq(0,1)$. 
\end{definition}
\noindent
We will denote $\lambda$ as permissible,  or $\mu^\pm_\lambda(x)$ as permissible, if $(\lambda, \mu^\pm_\lambda(x))$ is permissible, as defined above. 

\begin{proposition} \label{p2}   $\mu^\pm_\lambda(x)$ is permissible if and only if  $\mu^\pm_\lambda(x)$ is real and non-negative for all $x\in(a,b)$, though not identically zero.  
\end{proposition} 
    
\begin{proof} 
If $\mu^\pm_\lambda(x)$ is real, non-negative and not identically zero for $x\in(a,b)$ then it is immediately clear that it is permissible.
Conversely, let $\mu^\pm_\lambda(x)$ be permissible. 
Given  the square root in condition \eqref{eq:FundConstr0001} is the positive one, 
and working in the complex plane such that any argument, denoted $\theta$ below, is in the range $\theta \in [0,2\pi)$ then the square root of $z=r\exp(\mathrm{i}\theta),$~$r\geq0$, is given by 
 $  \sqrt{r}e^{i\theta/2}.$
Hence, any imaginary  contribution to 
$\sqrt{\mu^\pm_\lambda(x)}$, 
in condition (\ref{eq:FundConstr0001}) is non-negative as $\theta/2\in[0,\pi)$ and cannot be cancelled from elsewhere in the integration domain. 
Thus, given Eq.~(\ref{eq:FundConstr0001}), 
 $\sqrt{\mu^\pm_\lambda(x)}$ must be real for all $x\in(a,b)$. Hence, 
$\mu^\pm_\lambda(x)$ is real and non-negative for all $x\in(a,b)$, while $\mu^\pm_\lambda(x)$ cannot be identically zero as the  integral in Eq.~(\ref{eq:FundConstr0001}) is not zero. 
\end{proof} 

\begin{proposition} \label{p3}   Assume $\lambda$ is  permissible for $x\in(a,b)$. If $\lambda$ has a non-zero imaginary part, $\Im(\lambda)\neq 0$, then $\Re(\lambda)<0$ for $x\in(a,b)$. Equivalently, if  $ \Re(\lambda)\geq 0$ then $\Im(\lambda)=0$ for $x\in(a,b)$.
\end{proposition}

\begin{proof}   
From the definition of $\mu^\pm_\lambda(x)$, we have 
\beq\label{mul} \det [-\mu^\pm_\lambda(x)  \D  +  \J _\lambda(x)]= \det [-\mu^\pm_\lambda(x)  \D  +  \J - \lambda \bm I] =0, 
\eeq
and so, dropping the explicit $x$-dependence 
\beq \label{muleqn} 
2 \lambda &=& \tr(-\mu^\pm_\lambda  \D +\J)
\pm\sqrt{[\tr(-\mu^\pm_\lambda  \D +\J  )]^2-4 \det[-\mu^\pm_\lambda(x)  \D + \J] },
\eeq
with the spatial dependence of $\mu^\pm_\lambda(x)$ such that the growth rate, $\lambda$, does not have a dependence on $x$. We have tr$( \J )<0$ for 
all $x$ by \cref{hcds} and, 
given $\lambda$ is permissible for $x\in(a,b)$, so that $\mu^\pm_\lambda(x)$ is permissible on this region, $\tr(-\mu^\pm_\lambda(x)  \D)<0$ for $x\in(a,b)$. Thus, if a permissible $\lambda$ is complex, it must have a negative real part. 
\end{proof}


\begin{proposition}\label{p3.5}
Given $(\lambda, \mu^\pm_\lambda(x))$ is permissible for $x\in(a,b)$ and $\Re(\lambda)\geq 0$, then $\det(\Bl) >0, $ where $\Bl=\D^{-1}(\J-\lambda \bm I)= \D^{-1}\J_\lambda$. 
\end{proposition} 
\begin{proof}
We have that $\lambda$ is real by \cref{p3} and is thus non-negative, with 
$$ \det(\Bl)= \det(\D^{-1}) \det( \J_\lambda)  = \det(\D^{-1})\left( \lambda^2 - \lambda \tr(\J) + \det(\J)\right) >0$$
for all $x\in(a,b)$, where the final inequality arises from the positive definiteness of $\D$, and \cref{hcds}.
\end{proof}

\begin{proposition} \label{p4} Given $\Re(\lambda)\geq0$, the pair $(\lambda, \mu^\pm_\lambda(x))$ is permissible on $x\in(a,b)$ if and only if
\beq \label{tc233} \tr( \Bl) > 0, \,\,\, \,\,\, \,\,\, \,\,\, [\tr( \Bl{)}]^2-4 \det(\Bl) \geq 0,
\eeq 
for  $x \in(a,b)$. 
\end{proposition} 

\begin{proof}
We immediately have   
\beq   \label{ie3} \mu^\pm_\lambda(x) &=& \frac 12 \left[\tr( \Bl)
\pm\sqrt{[\tr( \Bl)]^2-4 \det(\Bl) }\right]
\eeq
since  $\mu^\pm_\lambda$ are defined to be the eigenvalues of 
$\Bl$.  
Given conditions \eqref{tc233}, we can see by Equation \eqref{ie3} that $\mu_\lambda^\pm(x) > 0$ for all $x \in (a,b)$, and, hence, condition \eqref{eq:FundConstr0001} is satisfied, giving permissibility.

Next we consider the converse by assuming  $(\lambda, \mu^\pm_\lambda(x))$ is permissible for $x\in(a,b)$. Then  $\Re(\lambda)\geq 0$ implies $\lambda$ is real by \cref{p3}. From permissibility and \cref{p2} we also have that $\mu^\pm_\lambda(x)$ is real, non-negative and not identically zero for   $x\in(a,b)$. As $\mu^\pm_\lambda(x)$ and $\lambda$ are real  this enforces
$$ [\tr( \Bl)]^2-4 \det(\Bl) \geq 0, $$ 
for  $x\in(a,b)$. 
Also $\det(\Bl) >0$ on this  region by \cref{p3.5} 
and hence for both the positive and negative square root in Equation \eqref{ie3}, the fact that $\mu^\pm_\lambda(x)$ cannot be negative enforces $\tr( \Bl) \geq 0$ for   $x\in(a,b)$. The possibility that $\tr( \Bl) = 0$ is excluded as then $\mu^\pm_\lambda(x)$ is not real, since  det$(\Bl)>0.$
 \end{proof}

Note that the conditions in \cref{p4} translate the constraint \eqref{eq:FundConstr0001} from the WKB solutions, $\bm w_\pm$ of \cref{w_sol}, to properties of $\Bl$ and thus properties of $\J, ~\D$ and $\lambda.$
Furthermore  the conditions in \cref{p4}
do not depend on the positive or negative branch of $\mu_\lambda^\pm$,   implying that both eigenvalues are permissible given that one of them is.

\subsection{Prospective Blow-up and Regularization of WKB Solutions} \label{pbu} \emph{A priori}, there is scope for the WKB solutions, $\bm w_\pm$ of \cref{w_sol}, to become unbounded, either within the region of permissibility $(a,b)$, or at its edges, $\{a,b\}$ due to a zero of  
the denominator of $2\bm s_{*\pm}^T\bm p_{*\pm}$ within the integrand of the exponent or at the edges due to 
denominator $(\mu_\lambda^\pm)^{1/4}$.

However, Propositions \ref{p3.5}, \ref{p4} show that   given $(\lambda, \mu^\pm_\lambda(x))$ is permissible and $\Re(\lambda)\geq 0$ on $x\in(a,b)$ the eigenvalues 
$\mu_\lambda^\pm$, as given by 
\beq   \label{ie33} \mu^\pm_\lambda(x) &=& \frac 12 \left[\tr( \Bl)
\pm\sqrt{[\tr( \Bl)]^2-4 \det(\Bl) }\right]
\eeq
are bounded away from zero on $x\in(a,b)$. Thus the denominator $(\mu_\lambda^\pm)^{1/4}$ is bounded away from zero on  $(a,b)$ and thus cannot generate blow up in the WKB solutions on approaching  either of the points $\{a,b\}$. 

However, there is still scope for the denominator of $2\bm s_{*\pm}^T\bm p_{*\pm}$ within the integrand of the exponent to generate blow-up. Thus we first determine conditions on $\bm B_\lambda$ for such a blow up to occur and  proceed to  demonstrate that it can be regularized, i.e.~bounded, on use of an internal homogeneous Dirichlet  boundary condition where the WKB solution would otherwise blow up. The resulting  non-trivial WKB solution will match a trivial zero solution  exterior to the region of permissibility (if the prospective blow up location is not on the edge of the domain $[0,1]$.) Below for notational convenience and brevity we drop the $\pm$ labels on  $\st, \ps$.

\begin{proposition}\label{nonorth}
 Given $(\lambda, \mu^\pm_\lambda(x))$ is permissible and $\Re(\lambda)\geq 0$ for $x\in(a,b)$, then $$[\tr( \Bl)]^2-4 \det(\Bl) > 0$$ for   $x \in (a,b)$ if and only if ${\bf s}_*^T {\bf p}_*\neq0$ for all $x \in (a,b)$, where ${\bf s}_*$ and ${\bf p}_*$ are the left and right unit eigenvectors of $[-\mu_\lambda^\pm{\bf I}+ {\bf B}_\lambda]$. 
\end{proposition}
\begin{proof}
We will demonstrate both implications via contraposition. We first assume that ${\bf s}_*^T {\bf p}_*=0$ at some point $x_*\in(a,b)$. By elaborating possibilities on a case by case basis for a general $2\times 2$ matrix with zero determinant, we note that the left and right eigenvectors of the zero eigenvalue can only be perpendicular if the matrix is proportional to one of the following:
$$ \left( \begin{array}{cc} 0 & 0 \\ 0 & 0 \end{array} \right) , 
~~~~~~~
\left( \begin{array}{rr} 1 & 1 \\ -1 & -1 \end{array} \right) ,~~~~~~~
\left( \begin{array}{rr} 1 & -1 \\ 1 & -1 \end{array} \right) . 
$$ 
In all three cases, we have that the trace is zero.
Therefore, $$\tr(-\mu_\lambda^\pm{\bf I}+ {\bf B}_\lambda) = -2\mu_\lambda^\pm+\tr({\bf B}_\lambda)=0.$$ However, by \cref{ie3}, this implies that $[\tr( \Bl)]^2-4 \det(\Bl) = 0$, contradicting the assumption that this quantity remains positive.

For the converse, we start from   $[\tr( \Bl)]^2-4 \det(\Bl) = 0$ at some point $x_*\in (a,b)$ (noting that if this term were negative, then, by \cref{p4}, $\lambda$ would not be permissible and we would have an immediate contradiction). By using Equation \eqref{ie3} again we see that $\tr(-\mu_\lambda^\pm{\bf I}+ {\bf B}_\lambda)=0,$ while $\rm{det}(-\mu_\lambda^\pm{\bf I}+ {\bf B}_\lambda)=0$ as $\mu_\lambda^\pm$ are the eigenvalues of ${\bf B}_\lambda$. Any real $2\times 2$ matrix with zero determinant and trace can be written in one of the following forms:
\begin{equation} \label{cases} \left( \begin{array}{cc} c_1 & c_2 \\ -\frac{c_1^2}{c_2} & -c_1 \end{array} \right) , ~~~~~~~
\left( \begin{array}{cc} 0 & 0 \\ c_2 & 0 \end{array} \right) , ~~~~~~~
\left( \begin{array}{cc} 0 & c_2 \\ 0& 0 \end{array} \right) , ~~~~~~~
\left( \begin{array}{cc} 0 & 0 \\ 0 & 0 \end{array} \right) , \, \, \, 
\end{equation}
for real $c_1\neq 0$ and real $c_2\neq 0$. The first of these has one left and one right eigenvector, given by ${\bf s}^T_* = (c_1,c_2)$ and ${\bf p}^T_* = (-c_2,c_1)$ to within normalization and which satisfy ${\bf s}_*^T{\bf p}_*=0.$ For the second matrix case, we have $\bm p_*^T = (0,1)$ and $\bm s_*^T = (1,0)$ which are also orthogonal, and similarly for the third case with $\bm p_*^T = (1,0)$ and $\bm s_*^T = (0,1)$. 

The final case is not possible given the constraints on the system. \alk{Specifically,} we have \alk{that} $ \J - \lambda \bm I = \bm \D {\bf B}_\lambda = \mu_\lambda^\pm \D,$ 
where we note that the constraint that $\mu_\lambda^\pm$ being permissible implies that   $\mu_\lambda^\pm$ is real and non-negative by \cref{p2}, and we have the restriction $\Re (\lambda)\geq 0$. Hence taking the trace of $\J=  \lambda \bm{I}+\mu_\lambda^\pm \D$ we have
$$ \mathrm{tr}( \J) = 2\lambda + \mu_\lambda^\pm  \mathrm{tr}(\D) \geq 0, $$ 
where we have used the fact that $\D$ is positive definite. The above inequality  contradicts the requirement of stability to homogeneous perturbations, \cref{hcds}.

\end{proof}

We proceed to consider whether there is a singularity when $${\bf s}_*^T {\bf p}_*=0 = [\tr( \Bl )]^2-4 \det(\Bl) $$ subject to weak constraints on the nature of the zero, which are given via the definition of an \emph{Admissible Neighborhood} immediately below. 

\begin{definition} {\bf Admissible Neighborhood.} \label{def9} Let $(\lambda, \mu^\pm_\lambda(x))$ be permissible for $x\in(a,b)$ and $\Re(\lambda)\geq 0$, with a simple zero  of $$ [\tr( \Bl )]^2-4 \det(\Bl) =0$$ 
 at a point $X_*\in\{a,b\}$. We can then write
  $$ [\mathrm{tr}( \Bl )]^2-4 \det(\Bl) = A^2|x-X_*| (1+o(1)),  \,  \,  \,  \,  \, \,  \,  \,  \,  \,  \,     \,   x\in \mathcal{N} $$ 
where $A>0$ without loss, given the {\it admissible  neighborhood} $\mathcal{N}$ is defined as follows.  If $X_*=a$ then
$$ [\tr( \Bl )]^2-4 \det(\Bl) >0$$ 
to the right of $X_*$ and 
$\mathcal{N}$ is the closure of the intersection of a sufficiently small, non-empty neighborhood of $X_*$ with the set $x\geq X_*=a$. In particular the neighborhood is sufficiently small to ensure no other zero of $[\mathrm{tr}( \Bl )]^2-4 \det(\Bl)$ is within $\mathcal{N}/\{X_*\}$.  If $X_*=b$, then  $\mathcal{N}$ is an analogously defined closure of a  half-neighborhood contained within $[a,b]$. \end{definition}

In particular we need  to consider how the solution behaves sufficiently close to a simple  zero. To proceed 
we first need to determine the behavior of $-\mu_\lambda^\pm{\bf I}+ {\bf B}_\lambda$  sufficiently close to a simple  zero, as summarized by the following proposition. 

\begin{proposition}\label{p8.55} 
Let $(\lambda, \mu^\pm_\lambda(x))$ be permissible and $\Re(\lambda)\geq 0$ for $x\in(a,b)$. Suppose 
 that a point $X_*\in\{a,b\}$ is a simple  zero of $[\tr( \Bl )]^2-4 \det(\Bl)$,
  with associated admissible neighborhood $\mathcal{N}$, as in \cref{def9}. Then for $x$ in $\mathcal{N}$ excluding the singularity point, that is 
 $x\in \mathcal{N}/ \{X_*\}$, we have 
  \beq \label{auxmatZ}   -\mu_\lambda^\pm{\bf I}+ {\bf B}_\lambda &=&  
 \left(
\begin{array}{cc}
c_1(x) & c_2(x)\\  \dfrac{c_1(x)}{c_2(x)}(-c_1(x)+K(x))  & -c_1(x)+K(x)
\end{array}\right), 
\eeq
where 
$$ K(x) =  \mathrm{tr}(-\mu_\lambda^\pm{\bf I}+ {\bf B}_\lambda)  = \mp A |x-X_*|^{1/2}(1+o(1)), \,\,\,\,\,\,\, A>0, $$ with the $\mp$ inherited from the $\pm$ of $\mu_\lambda^\pm$ and $c_2(x)$  is  non-zero for $x\in\mathcal{N}, $ with restrictions on $\mathcal{N}$ as required.  
\end{proposition}

\begin{proof} 
We first of  all note  that from \cref{ie33} and \cref{def9} that 
\beq  K(x) &:=& \tr(-\mu_\lambda^\pm{\bf I}+ {\bf B}_\lambda)  =   \mp A |x-X_*|^{1/2}(1+o(1)), 
\eeq 
where $A>0$.  Noting that 
$$  \det(-\mu_\lambda^\pm{\bf I}+ {\bf B}_\lambda)=0, \,\,\,\,\, x\in \mathcal{N}  $$ 
then, with the zero determinant  constraint  $\alpha(x)(K(x)-\alpha(x))=\beta(x)\gamma(x)$, we can write, without loss of generality, that 
\beq \label{al1}  -\mu_\lambda^\pm{\bf I}+ {\bf B}_\lambda = \left(
\begin{array}{cc}
\alpha(x) & \beta(x) \\  \gamma(x)  & K(x)-\alpha(x)
\end{array}\right).
\eeq

Firstly, if $\beta(X_*)\neq 0$ then restricting $\mathcal{N}$ as necessary, we have $\beta(x)\neq 0$ for $x\in\mathcal{N}.$ In this case we can write  
 \beq  -\mu_\lambda^\pm{\bf I}+ {\bf B}_\lambda &=&  
 \left(
\begin{array}{cc}
c_1(x) & c_2(x)\\  \dfrac{c_1(x)}{c_2(x)}(-c_1(x)+K(x))  & -c_1(x)+K(x)
\end{array}\right), 
\eeq
where $c_2(x)$ is  non-zero for $x\in\mathcal{N}.$ Hence the proposition  holds if $\beta(X_*)\neq 0$. 

 We now consider the degeneracy with $\beta(X_*)=0$, which implies  $\alpha(X_*)=0$ on noting that $K(X_*)=0.$ From \cref{al1} and \cref{ie33}, together with the smoothness of $\bm B_\lambda$ allowing its expansion about $x=x_*$,  we have 
 \beq\label{muha} \alpha(x) &=  -\mu_\lambda^\pm + (\bm B_\lambda)_{11}(x) \\ \nonumber &=  \frac 1 2 ( (\bm B_\lambda)_{11}(X_*) - (\bm B_\lambda)_{22}(X_*))    \mp \frac 1 2 A |x-X_*|^{1/2}(1+o(1))  +O(|x-X_*|)      
\eeq
and $\alpha(X_*)=0$ implies $(\bm B_\lambda)_{11}(X_*) = (\bm B_\lambda)_{22}(X_*).$ Hence, noting that any $O(|x-X_*|)$ terms can be absorbed into the $|x-X_*|^{1/2}o(1)$ terms, we have 
$$ K(x)-\alpha(x) =  \mp \frac 1 2 A |x-X_*|^{1/2}(1+o(1)), $$
and thus 
$$\beta(x)\gamma(x) = \alpha(x)(K(x)-\alpha(x)) = \frac 1 4 A^2 |x-X_*|(1+o(1)).$$ 
Hence for $x\in \mathcal{N}/\{X_*\}$, restricting $\mathcal{N}$ as necessary to ensure the validity of the above local expansions,  we  have $\beta(x)\neq 0.$ Hence for $x\in \mathcal{N}/\{X_*\}$, \cref{auxmatZ} applies, demonstrating the proposition for $\beta(X_*)=0.$
\end{proof}

   The following  propositions proceed to consider how
  $$ \exp \left(\int  \frac{\bm{s}_{\bm *}(\bar{x})^T\D^{-1}(\bar{x})\bm{M}(\bar{x})\bm{p}_{\bm *}(\bar{x})+2\bm{s}_{\bm *}(\bar{x})^T\bm{p}_{\bm *}'(\bar{x})}{2{\bf s_*}(\bar{x})\cdot{\bf p}_*(\bar{x})} \mathrm{d}\bar{x} \right)$$ 
   behaves for appropriate integration limits sufficiently close to {the simple zero $X_*$,} before  returning to considering the WKB solution of \cref{w_sol}.

\begin{proposition} \label{p9} 
 Let $(\lambda, \mu^\pm_\lambda(x))$ be permissible and $\Re(\lambda)\geq 0$ for $x\in(a,b)$. Suppose 
 that a point $X_*\in\{a,b\}$ is a simple  zero of $[\tr( \Bl )]^2-4 \det(\Bl)$,
  with associated admissible neighborhood $\mathcal{N}$, as in \cref{def9}.   Then, if $X_*=a$, and with $x\in \mathcal{N}, ~x>X_*=a$, we have 
\beq \label{sing1} \exp \left(\int_{X_*+\eta}^{x} \frac{\bm{s}_{\bm *}(\bar{x})^T\D^{-1}(\bar{x})\bm{M}(\bar{x})\bm{p}_{\bm *}(\bar{x})+2\bm{s}_{\bm *}(\bar{x})^T\bm{p}_{\bm *}'(\bar{x})}{2{\bf s_*}(\bar{x})\cdot{\bf p}_*(\bar{x})} \mathrm{d}\bar{x} \right) = \mathrm{ord}\left(\frac 1 {\eta^{1/4}}\right),
 \eeq 
as $  \eta\rightarrow 0^+. $
If instead $X_*=b$  with $x\in \mathcal{N}$, i.e.~$x<X_*=b$, we have (noting the sign difference in the exponent)
 \beq\label{sing2} \exp \left(-\int^{X_*-\eta}_{x} \frac{\bm{s}_{\bm *}(\bar{x})^T\D^{-1}(\bar{x})\bm{M}(\bar{x})\bm{p}_{\bm *}(\bar{x})+2\bm{s}_{\bm *}(\bar{x})^T\bm{p}_{\bm *}'(\bar{x})}{2{\bf s_*}(\bar{x})\cdot{\bf p}_*(\bar{x})} \mathrm{d}\bar{x} \right) = \mathrm{ord}\left(\frac 1 {\eta^{1/4}}\right),
 \eeq 
as $  \eta\rightarrow 0^+. $ 
\end{proposition}

\begin{proof}  By \cref{p8.55} we have  
\beq \label{auxmat}  -\mu_\lambda^\pm\bm{I}+ &\Bl \\=&
\begin{pmatrix}
c_1(x) & c_2(x) \\  \frac{c_1(x)}{c_2(x)}\left(-c_1(x) \mp  A |x-X_*|^{1/2}(1+o(1))\right)  & -c_1(x) \mp A |x-X_*|^{1/2}(1+o(1))
\end{pmatrix}, 
\eeq
for $x\in \mathcal{N}/\{X_*\}$, where   $c_2(x)\neq 0$.  

Noting   overall sign choices are  without loss of generality due to the parity of the integrands in Eqs.~\eqref{sing1},~\eqref{sing2}, we have that the unit right zero eigenvector of the matrix \eqref{auxmat} can be written as 
$$ \bm p_* = \frac 1 {N_1} \left(\begin{array}{c}  -c_2 \\ ~~c_1\end{array}\right),  \,\,\,\,\,\, N_1 = (c_1^2+c_2^2)^{1/2}\geq|c_2(x)|>0,$$
where $N_1 $ is a normalization factor. Similarly the left unit eigenvector of matrix \eqref{auxmat} is 
$$ \bm s^T_* = \frac 1 {N_2} \left( \frac 1 {c_2}\left(  c_1 \pm A |x-X_*|^{1/2}(1+o(1)   \right), ~1 \right), $$
where $N_2\geq 1 $ is a normalization factor. 

Further, from the first row of \cref{auxmat},
we have 
$$ c_1(x) = ({\bf B}_\lambda)_{11} -\frac 12 \tr({\bf B}_\lambda) (X_*) \mp \frac A 2 |x-X_*|^{1/2}(1+o(1)), \, \, \, \, c_2(x) = ({\bf B}_\lambda)_{12},$$ 
where $({\bf B}_\lambda)_{11}, ({\bf B}_\lambda)_{12}$ are components of ${\bf B}_\lambda$ and thus smooth under differentiation with respect to $x$. Thus $c_2'(x)$ is bounded on $x\in \mathcal{N}/\{X_*\}$, while 
$$ c_1'(x) =  \mp \frac A 4  \sigma_{X_*}  |x-X_*|^{-1/2}(1+o(1)), \,\,\,\,\,\,\,\,\,\,\,\,
\sigma_{X_*} := \left\{ \begin{array}{rl} 
\,1 & X_*=a \\ -1 & X_*=b\end{array} \right. .$$

To determine a leading order approximation for $\bm p_*'$, the derivative of $\bm p_*$, note that $c_1'(x)$ will have a large derivative in $\mathcal{N}$, giving 
$$ \bm p_*' = \frac 1 {N_1} \left(\begin{array}{c}  0 \\  \mp \frac A 4  \sigma_{X_*}  |x-X_*|^{-1/2}\end{array}\right)(1+o(1)),
$$
as all other derivatives are ord$(1)$, except for the derivative of $c_1$ in the denominator, that is $N_1$. As this is in the denominator however, it will suppress, rather than elevate the asymptotic order of the resulting term. 
{Further note that  $ \bm{s}_{\bm *}(\bar{x}), $  $\bm{p}_{\bm *}(\bar{x}) 
$ are normalized  by construction and thus have components that are $O(1)$ on approaching the singular point. Similarly,  from our assumptions that $\bm D$, $\bm D^{-1}$ and $\bm J$ have smooth and bounded coefficients we also have $\bm{M}(\bar{x})$    has bounded coefficients, that is the coefficients are $O(1)$   on approaching the simple zero, $X_*$. Hence   we  have }
$$ \bm{s}_{\bm *}(\bar{x})^T\D^{-1}(\bar{x})\bm{M}(\bar{x})\bm{p}_{\bm *}(\bar{x}) \sim { O(1)}$$ {and thus}
\beq 
\frac{\bm{s}_{\bm *}(\bar{x})^T\D^{-1}(\bar{x})\bm{M}(\bar{x})\bm{p}_{\bm *}(\bar{x})+2\bm{s}_{\bm *}(\bar{x})^T\bm{p}_{\bm *}'(\bar{x})}{2{\bf s_*}(\bar{x})\cdot{\bf p}_*(\bar{x})} &= 
\frac{\bm{s}_{\bm *}(\bar{x})^T\bm{p}_{\bm *}'(\bar{x})}{{\bf s_*}(\bar{x})\cdot{\bf p}_*(\bar{x})}(1+o(1))
\\ &= \frac 1 4 \sigma_{X_*} \frac 1 {|x-X_*|}(1+o(1))  \, \, \, \, \, \, \, \, \,   \nonumber 
\eeq 
for $x\in \mathcal{N}/\{X_*\}$ {since the dominant contribution in the numerator of the latter relation is from the singularity term scaling with $|x-X_*|^{-1/2}$ within $\bm p'_*$, noting that this vector is {\it not} normalized.}

Hence, fixing  $\eta >0$ sufficiently small, with  $X_*=a$, $x\geq X_*+\eta, ~ x\in \mathcal{N}/\{X_*\}$ we have that the left-hand side of \cref{sing1} reduces to 
\beq \label{singint1}
\exp \left(\int_{X_*+\eta}^{x} \frac 1 {4(\bar{x}-X_*)}(1+o(1)) \mathrm{d}\bar{x} \right) \sim\mathrm{ord}\left(\frac 1 {\eta^{1/4}}\right) 
 , \, \,     
\eeq
 providing the required result on subsequently considering $\eta \rightarrow 0^+.$ Analogously, for sufficiently small, fixed, $\eta >0$ with $X_*=b$, $x\leq X_*-\eta, ~ x\in \mathcal{N}/\{X_*\}$ we have the left-hand side of \cref{sing2} reduces to  \beq \label{singint2}
\exp \left(\int^{X_*-\eta}_{x} \frac 1 {4(X_*-\bar{x})}(1+o(1)) \mathrm{d}\bar{x} \right) \sim\mathrm{ord}\left(\frac 1 {\eta^{1/4}}\right) 
 , \, \,     
\eeq giving the required result on now taking $\eta\rightarrow 0^+.$  
\end{proof}

We can use the results of the above proposition to demonstrate that the behavior of the WKB solution is singular  as
$ [\tr( \Bl )]^2-4 \det(\Bl) $ approaches zero from above, but that the singularity can be regularized with an appropriate  Dirichlet  boundary condition.  Firstly, note the WKB solutions are only defined up to a multiplicative scaling so that on a region $(a,b)$ with $0\leq a< b\leq 1$, we need to assign an overall scale, which we impose at the center of the domain 
$ X_{**} =  (a+b)/2$  by setting 
$$ Q_0^\pm(X_{**}) = \mathrm{ord}(1).$$ Thus for $\eta >0$ sufficiently small we are interested in the behavior of $Q_0$, normalized by {its value at $X_{**}$} as this will dictate the behavior of $\bm w_\pm$, as we now characterize in the following proposition.

\begin{proposition} \label{p10}
Let $(\lambda, \mu^\pm_\lambda(x))$ be permissible and $\Re(\lambda)\geq 0$ for $x\in(a,b)$. Suppose 
 that a point $X_*\in\{a,b\}$ is a simple  zero of $[\tr( \Bl )]^2-4 \det(\Bl)$,
  with associated admissible neighborhood $\mathcal{N}$, as in \cref{def9}. 
Defining  $X_{**}:=(a+b)/2$, $\sigma_{X_*}=1$ if $X_*=a$ and $\sigma_{X_*}=-1$ if $X_*=b$ we have
 $$ \frac{Q_0^\pm(X_*+\sigma_{X_*}\eta)}{Q_0^\pm(X_{**})} \sim \mathrm{ord}\left(\frac 1 {\eta^{1/4}}\right)$$ on   $x\in \mathcal{N}/\{X_*\}$ as $\eta\rightarrow 0^+$, so that from \cref{w_sol} 
 $$ \bm w_\pm (X_*+\sigma_{X_*}\eta) 
= \mathrm{ord} (\eta^{3/4})$$ on $x\in \mathcal{N}/\{X_*\}$ as $\eta\rightarrow 0^+$, 
 providing homogeneous Dirichlet  conditions are imposed at $x=X_*$.
\end{proposition}

\begin{proof}
Let $X_{**}=(a+b)/2$, fix $Q^\pm_0(X_{**})=\mathrm{ord}(1)$ for the overall scaling of the WKB solution and, for convenience, define 
$$ I(\bar{x}) := \frac{\bm{s}_{\bm *}(\bar{x})^T\D^{-1}(\bar{x})\bm{M}(\bar{x})\bm{p}_{\bm *}(\bar{x})+2\bm{s}_{\bm *}(\bar{x})^T\bm{p}_{\bm *}'(\bar{x})}{2{\bf s_*}(\bar{x})\cdot{\bf p}_*(\bar{x})} .$$
Then, noting \cref{q0eqn}, if $X_*=a$ we have 
\beq  \label{Q0a} \frac{Q_0^\pm(X_*+\eta)}{Q_0^\pm(X_{**})} &=&
 \left(  \frac{\mu_\lambda^\pm(X_{**})}{\mu_\lambda^\pm(X_*+\eta)} \right)^{1/4}  
   \exp \left(\int_{X_*+\eta}^{X_{**}}  I(\bar{x}) \mathrm{d}\bar{x} \right)  \sim \mathrm{ord}\left(\frac 1 {{ \eta^{1/4}}}\right). 
 \,\,\,\,
\, \, 
\eeq 
The final observation arises from using a few previous results in combination. Firstly, we have that $\mu_\lambda^\pm$ is bounded away from zero by \cref{p4} and we recall that  $\mu_\lambda^\pm$ is an eigenvalue of $\bm B_\lambda$, given by \cref{ie33}. Then \cref{p9}, which also applies when  $X_*=b$, gives the final result:
\beq  \frac{Q_0^\pm(X_*-\eta)}{Q_0^\pm(X_{**})} =
 \left(  \frac{\mu_\lambda^\pm(X_{**})}{\mu_\lambda^\pm(X_*-\eta)} \right)^{1/4}  
   \exp \left(-\int^{X_*-\eta}_{X_{**}}  I(\bar{x}) \mathrm{d}\bar{x} \right)  \sim \mathrm{ord}\left(\frac 1 {\eta^{1/4}}\right). \,\,\,\,
\, \, \label{eq54}
\eeq 
\alk{To match} Dirichlet  boundary conditions at $x=X_*$ we have that at the point $x=X_*+\sigma_{X_*} \eta$  that \alk{(using the arbitrary constants $S_0^\pm$ and $C_0^\pm$)} double angle formulae can be used to rewrite the trigonometric contribution to $\bm w_\pm$  within the expression of \cref{w_sol}  in the form
\beq \label{triga} \alk{\widehat{S_0^\pm}}\sin\left(\frac{1}{\ep}\int_{X_*}^{x=X_*+\eta} \sqrt{\mu_\lambda^\pm(\bar{x})}d\bar{x} \right)  &=& \mathrm{ord}\left( \frac \eta \ep\right)= \mathrm{ord}\left( \eta \right), \,\,\,\,\,\,\,\, X_*=a  \\ \alk{\widehat{S_0^\pm}}\sin\left(\frac{1}{\ep}\int^{X_*}_{x=X_*-\eta} \sqrt{\mu_\lambda^\pm(\bar{x})}d\bar{x} \right)  &=&  \mathrm{ord}\left( \frac \eta \ep\right)= \mathrm{ord}\left( \eta \right), \,\,\,\,\,\,\,\, X_*=b 
\eeq
as $\eta\rightarrow 0^+$, noting as above that $\mu_\lambda^\pm$ is bounded away from zero\alk{, and where \alk{$\widehat{S_0^\pm}$} is a real constant}. {In particular  the analogous cosine contribution is  set to} zero by the imposition of the Dirichlet  condition at $x=X_*$ and  $\alk{\widehat{S_0^\pm}} \neq 0$ to avoid the trivial solution.  Eqs.~\eqref{Q0a},\eqref{eq54}, and \eqref{triga} then give  
$$ \bm w_\pm (X_*+\sigma_{X_*}\eta) = \mathrm{ord}(\eta^{3/4}).$$ 
\end{proof}
An immediate corollary of \cref{p10} is that if $a=0$ or $b=1$ (but not both), then we need to apply both a homogeneous Dirichlet  condition to maintain boundedness, as above, and homogeneous Neumann conditions to satisfy the zero flux boundary conditions at the edge of the domain. Only the trivial solution is then possible, so we have $a=0$ and $b=1$ are excluded from this class of solutions. 

We further note that dealing with the case $X_*=b$ in addition to $X_*=a$ adds to the detail of the calculations required to deduce the above results, though not the concepts. One could alternatively deduce the result for $X_*=b$ with the approach to the singular point from below, from the result with $X_*=a$ with the approach to the singular point from above as follows. In particular 
with $$  x\rightarrow 1-x , \, \,  \J(x)\rightarrow \J(1-x), \, \,  \D(x)\rightarrow \D(1-x), $$ and noting $\bm M(x), \bm N(x)$ remain smooth and  bounded under this mapping, a singularity approached from below at $X_*=b$ 
will be mapped to a singularity approached from above at $X_*=1-b$, but the scale of the singularity will be invariant. The latter will 
have $\bm w_\pm$ scale with $\eta^{1/4}$, where $x=1-b+\eta$, and hence the singularity approached from below at $X_*=b$ of the original problem will 
also scale with $\eta^{1/4}$. Thus, one may alternatively  only track the calculations for $X_*=a$ above and then infer  the result for $X_*=b$ using this reasoning.

In addition, the following Proposition also demonstrates that if these conditions hold for real $\lambda_* > 0$, then they hold for all non-negative {$\lambda < \lambda_*$.}

\begin{proposition} \label{newp} Assume the homogeneous stability constraints of \cref{hcds} are given. Then, if $\lambda_*>0$ is permissible for  $x \in (a,b)$, all $\lambda \in [0,\lambda_*]$  are permissible for $x \in (a,b)$.
\end{proposition} 

\begin{proof}  
  Recall that $\bm B_\lambda = \D^{-1}(\J - \lambda \bm I)$ and hence 
$$ \mathrm{tr}(\bm B_\lambda)  = \mathrm{tr}(\D^{-1} \J ) - \lambda \mathrm{tr}(\D^{-1} ), \, \, \, \, \, \, \, \, 
  \mathrm{det}(\bm B_\lambda)  =\mathrm{det}(\D^{-1})(\lambda^2 - \lambda (\mathrm{tr} \J) +  \mathrm{det} (\J)) .$$ 
  Thus, for $\lambda_*>0$ permissible with $\lambda_*> \lambda\geq 0$ we have 
  $$ \mathrm{tr}(\bm B_\lambda) = \mathrm{tr}(\bm B_{\lambda_*}) +(\lambda_*-\lambda) \mathrm{tr}(\D^{-1} ) > \mathrm{tr}(\bm B_{\lambda_*}) > 0, $$ 
  noting $\lambda_*$ is permissible and so 
  $\mathrm{tr}(\bm B_{\lambda_*}) >0$ and 
  $\mathrm{tr}(\D^{-1} )>0$ by positive 
  definiteness of $\D.$ Hence the first 
  condition of \cref{tc233}  is satisfied for $\bm B_\lambda$. 
  
  We also need to demonstrate the second condition of \cref{tc233}  to demonstrate  permissibility by \cref{p4}. 
Hence  we define 
  $$ P(\lambda) := \mathrm{tr}(\bm B_\lambda))^2-4\mathrm{det}(\bm B_\lambda) = \zeta_2 \lambda^2 + \zeta_1\lambda + \zeta_0,$$
  where 
 \beq \nonumber 
 &\zeta_2 = (\mathrm{tr}(\D^{-1}))^2-4\mathrm{det}(\D^{-1}),  \,\,  \,\,  \,
 \zeta_1 = 4 \mathrm{det}(\D^{-1})\mathrm{tr}(\bm  J )
 - 2\mathrm{tr}(\D^{-1} \J )\mathrm{tr}(\D^{-1}  ), \\
 & \textrm{and }\zeta_0 =  
 (\mathrm{tr}(\D^{-1}\J ))^2
- 4\mathrm{det}(\D^{-1}\J).  
 \eeq 
With $\eta_1,\,\eta_2$ the eigenvalues of $\D^{-1}$, which are real and positive  by positive definiteness of $\D$, but  conceivably  repeated, note that 
$$ \zeta_2 = (\eta_1-\eta_2)^2 \geq 0.$$ Furthermore, let 
$$ \lambda_1 = \frac{\mathrm{tr}(\D^{-1}  \J )}{\mathrm{tr}(\D^{-1}  )}, $$ so that $\mathrm{tr}(\bm B_{\lambda_1}  )=0$. Thus  $\lambda_2 \geq  \lambda_1$ implies  $\mathrm{tr}(\bm B_{\lambda_2}  ) \leq 0$, which in turn implies  that  $\lambda_2$ is not permissible by \cref{p4}. Hence $\lambda_*$ permissible gives $\lambda_*<\lambda_1$.  Also, as $\lambda_* >0$, we then have $\lambda_*^2<\lambda_1^2$ and hence
\begin{align*}P(\lambda_*)=&P(\lambda_1) \alk{+(\mathrm{tr}(\D^{-1}\lambda_*-\mathrm{tr}(\D^{-1}\J))^2} +4\mathrm{det}(\D^{-1})( \lambda_1^2-\lambda_*^2) \\& \alk{-4\mathrm{det}(\D^{-1})}\mathrm{tr}( \J )(\lambda_1-\lambda_*) >P(\lambda_1),
\end{align*}
using the homogeneous conditions of \cref{hcds} and the positive definiteness of $\D$. Thus, if $\zeta_2>0$, then $\lambda_*$ is on the decreasing branch of the quadratic $P(\lambda)$, {as is any smaller value of $\lambda$ since $P(\xi) \rightarrow \infty$ as $\xi\rightarrow -\infty$. Alternatively,} if $\zeta_2=0$, 
we have the degenerate linear case for $P(\lambda)$ with a negative gradient. Either way, for $0\leq  \lambda < \lambda_*$ we have 
$$P(\lambda) > P(\lambda_*) \geq 0$$
and the second condition  \cref{tc233}  is satisfied for $\bm B_\lambda$ and we have  permissibility by \cref{p4}. 
 \end{proof}

 \subsection{Localized WKB Solutions}
  
We restrict ourselves to non-negative growth rates,  $\lambda,$ {so that $\Re(\lambda)>0$, while recalling that $\lambda$ does not vary across the domain and must be such that $\mu^\pm_\lambda$ satisfies} the constraint \eqref{eq:FundConstr0001} and thus be permissible. Hence, by \cref{p3} we further have that $\lambda$ is real.  
Given the equivalent conditions of permissibility of Eq.\eqref{tc233} from \cref{p4},  \cref{p10} shows the behavior of the WKB leading order solutions near singular points, while \cref{nonorth} shows how to locate singular points for a given non-negative permissible growth rate, $\lambda$, from $\D, \J.$  These propositions, and thus the form of the resulting solutions, require that 
\beq \label{saz1}  [\tr({\bf B}_\lambda)]^2-4 \det({\bf B}_\lambda) 
\eeq 
only has simple  zeros, if any. 

In general we have considered open intervals $(a,b) \subseteq(0,1)$ where WKB leading order solutions for such $\lambda$ exist, though in general the set on which they exist can be a union of such domains, and all of the above propositions apply. Thus we use  the definition:
\begin{definition}
For $\lambda>0$   we define $\mathcal{T}_\lambda$ to be the closure of the maximal open set where $\lambda$ is permissible. 
\end{definition}
In particular $\mathcal{T}_\lambda$ need not be simply connected. At its edges we also have either homogeneous Neumann conditions from the boundary conditions at $x\in\{0,1\}$ or homogeneous Dirichlet  conditions from \cref{p10} for internal edges, assuming these {correspond to} simple   zeros of \cref{saz1}, noting a bounded non-trivial  WKB leading order solution of the form of \cref{w_sol} is required for the solution to represent an unstable mode. Note that no singularities emerge from the denominator $[\mu_\lambda^\pm]^{1/4}$, as discussed at the start of \cref{pbu}. 
In addition, a final wave selection condition must be imposed, {such as \cref{eq:FundConstr0001} across one of the intervals constituting  $\mathcal{T}_\lambda$ to ensure  the homogeneous edge conditions at the edges of this interval 
can be satisfied (for a suitable choice of sufficiently small  $\ep$ and thus the diffusion scale). The WKB solution on all other intervals constituting $\mathcal{T}_\lambda$ can then be set to the trivial zero solution. In addition, 
outside $\mathcal{T}_\lambda$ }  the lack of permissibility for $\lambda$ entails non-trivial WKB leading order solutions cannot satisfy the waveform selection constraints required to fulfil the homogeneous conditions  at the edges,  $\mathcal{T}_\lambda$,   leaving only the trivial WKB solution for the leading order solution in the complement of $\mathcal{T}_\lambda$. 

Strictly, {such solutions are only outer solutions} and an inner solution would be required to generate a smooth leading order composite solution. However, the outer solutions are sufficient for our purposes.

Our results thus far, for the simple case that $\mathcal{T}_\lambda$ is simply connected,  may be represented by the following theorem for instability, directly analogous to that of \cite{krause_WKB}, but now applicable to the more general cross-diffusion systems of \cref{orig_eqn}-\eqref{BCs}, linearizing to  \cref{linear_RDS}.

\begin{theorem}[$\lambda$-Dependent Heterogeneous Case]\label{heteroturing} Let $\Re(\lambda) > 0$, $0<\ep \ll 1$, and assume that   $$[\tr( \Bl (x))]^2-4 \det( \Bl (x))$$ has no more than two simple  zeros for $x \in [0,1]$, and is positive between these two zeros.  We assume stability to perturbations in the absence of diffusion, i.e.,
\beq \tr(  \J (x)) < 0, ~~ \det( \J (x)) > 0, ~~~ \text{for all }x \in [0,1],
\eeq
is given. Then there exists  non-homogeneous, non-trivial, bounded  perturbations $\bm w$ satisfying \cref{linear_RDS} 
to leading order in $\ep$, which grow as $e^{\lambda t}$ 
in the interval 
$x \in (a(\lambda),b(\lambda))$ if
\beq \label{tchet00}  \tr( {\bf B}_\lambda(x)) > 0, ~~ [\tr( {\bf B}_\lambda(x))]^2-4 \det({\bf B}_\lambda(x)) > 0, ~~~ \text{for all }x \in (a(\lambda),b(\lambda)), \eeq
and $\ep$ satisfies the wave-selection constraint 
\beq\label{hetselect}  \left(n^{\pm}+ \frac{K}{2}\right)\pi \ep = \int_{a(\lambda)}^{b(\lambda)} \sqrt{\mu_{\lambda}^\pm(\bar{x})}\mathrm{d}\bar{x} 
\eeq 
for any natural number $n^\pm>0$, such that
  $$a(\lambda) = \max(0, \min(\{x: [\tr( {\bf B}_\lambda(x))]^2-4 \det({\bf B}_\lambda(x)) =0\})), $$ $$ b(\lambda) = \min(1, \max(\{x: [\tr( {\bf B}_\lambda(x))]^2-4 \det({\bf B}_\lambda(x)) =0\})),$$ and $K=0$ if either $a(\lambda) = 0$ and $b(\lambda)=1,$ or if  $0 <a(\lambda) <b(\lambda)<1;$ otherwise $K=1$. 
\end{theorem}
\begin{proof}
We have $\mathcal{T}_\lambda =(a(\lambda), b(\lambda))$. Also  $\mu_\lambda^\pm$ is real and permissible {on $\mathcal{T}_\lambda$ by Propositions \ref{p2}, \ref{p4} while  \cref{ie33} and  conditions \eqref{tchet00} also give $\mu_\lambda^\pm$ is positive on $\mathcal{T}_\lambda$. 
Further, by \cref{p3}, we have no loss of generality in specializing to strictly real $\lambda$. Finally, by} \cref{p10}, we have that the WKB solutions can be bounded at any interior edges of $\mathcal{T}_\lambda$ by homogeneous Dirichlet conditions at the interior edges of $\mathcal{T}_\lambda$, and the boundary conditions require homogeneous Neumann conditions. 

The permissibility of $\mu_\lambda^\pm$ on $\mathcal{T}_\lambda$ allows 
non-trivial leading-order WKB solutions on $\mathcal{T}_\lambda$. Conversely, noting the definition of  permissibility, its absence  on the complement of $\mathcal{T}_\lambda$ implies that the wave selection constraint cannot be satisfied and thus there is only the trivial WKB solution on the complement of $\mathcal{T}_\lambda$. Thus we have the following solutions, classified by the possible forms of $\mathcal{T}_\lambda$,  given  $$[\tr( \Bl (x))]^2-4 \det( \Bl (x))$$ has no more than two simple  zeros for $x \in [0,1]$: 

\begin{itemize}
\item no singular points, so $\mathcal{T}_\lambda=[0,1]$ and the solution is
  \begin{subequations} 
\beq 
{\bf w}^\pm(x,t)  &= e^{\lambda t} \exp\left(-\int_{0}^x\frac{\bm{s}_{\bm *}^\pm(\bar{x})^T\D^{-1}(\bar{x})\bm{M}(\bar{x})\ps(\bar{x})+2\bm{s}_{\bm *}^\pm(\bar{x})^T(\ps)'(\bar{x})}{2\bm{s}_{\bm *}^\pm(\bar{x})^T\ps(\bar{x})}d\bar{x}\right) \nonumber 
\\ \label{eq:PiecewiseWKBsolu} &\times
\frac{C_{0}^{{\pm}}}{[\mu_\lambda^\pm(x)]^{1/4}} 
\cos\left(\frac 1\ep  \int_0^x \sqrt{\mu_\lambda^\pm(\bar{x})}\mathrm{d}\bar{x}\right) 
{\bf p}_*(x), \eeq 
with 
\beq
\int_0^1 \sqrt{\mu_\lambda^\pm(\bar{x})}  \mathrm{d}\bar{x} \,\,\,\,  &=&    n^\pm\pi\ep;
\eeq
\item one singular point $x_*(\lambda)>0$, so without loss of generality, $(x_*,1)= \mathcal{T}_\lambda$, with solution
  \beq
    {\bf w}^\pm(x,t)    &= e^{\lambda t} \exp\left(\int_{x}^1\frac{\bm{s}_{\bm *}^\pm(\bar{x})^T\D^{-1}(\bar{x})\bm{M}(\bar{x})\ps(\bar{x})+2\bm{s}_{\bm *}^\pm(\bar{x})^T(\ps)'(\bar{x})}{2\bm{s}_{\bm *}^\pm(\bar{x})^T\ps(\bar{x})}d\bar{x}\right) \hspace*{6mm} \nonumber 
\\ 
    \label{eq:PiecewiseWKBsolu_right} &\times  \frac{S_{0}^{{ \pm}}}{[\mu_{\lambda}^\pm(x)]^{1/4}} 
\sin\left(\frac 1\ep  \int_{x_*}^x \sqrt{\mu_{\lambda }^\pm(\bar{x})}\mathrm{d}\bar{x}\right) 
{\bf p}_{*\pm}(x), 
\eeq 
with 
\beq
 \int_{x_*}^1 \sqrt{\mu_{\lambda }^\pm(\bar{x})}\mathrm{d}\bar{x} \,\,\,\, &=&   \left(n^\pm+\frac{1}{2}\right)\pi\ep,
\eeq
for $x \in \mathcal{T}_\lambda$, and zero otherwise;
\item
  two singular points $x_*(\lambda),~x_{**}(\lambda)\in(0,1)$ delimiting the $\mathcal{T}_\lambda$ set, i.e. $\mathcal{T}_\lambda=(x_*,~x_{**})$, with solution
  \beq\nonumber 
    {\bf w}^\pm(x,t) &=  e^{\lambda t} \exp\left(\int_{x_*(\lambda)}^{x_{**}(\lambda)}\frac{\bm{s}_{\bm *}^\pm(\bar{x})^T\D^{-1}(\bar{x})\bm{M}(\bar{x})\ps(\bar{x})+2\bm{s}_{\bm *}^\pm(\bar{x})^T(\ps)'(\bar{x})}{2\bm{s}_{\bm *}^\pm(\bar{x})^T\ps(\bar{x})}d\bar{x}\right)
    \\   \label{eq:PiecewiseWKBsolu_interval} &\times 
 \frac{S_{0}^{{\pm}}}{[\mu_{\lambda}^\pm(x)]^{1/4}}
\sin\left(\frac 1\ep  \int_x^{x_{**}} \sqrt{\mu_{\lambda}^\pm(\bar{x})}\mathrm{d}\bar{x}\right) 
{\bf p}_{*\pm}, 
\eeq 
with 
\beq
 \int_{x_*(\lambda)}^{x_{**}(\lambda)} \sqrt{\mu_{\lambda }^\pm(\bar{x})}\mathrm{d}\bar{x} \,\,\,\, &=&  n^\pm \pi\ep,
\eeq
for $x \in \mathcal{T}_\lambda$, and zero otherwise. \label{eq:PiecewiseWKBsolu_full}
\end{subequations}
\end{itemize}
We note that requiring $\ep$ to satisfy the wave selection constraint ensures that the boundary conditions are satisfied, and hence  the above expressions constructively give the non-trivial and bounded  solutions required for the proof. 
\end{proof}

{The generalization of \cref{heteroturing} to more than two  singular points  and non-simply connected $\mathcal{T}_\lambda$ is directly analogous, except for the wave selection criteria.  One has two degrees of freedom, $\lambda$ and $\ep$, so it should be feasible to simultaneously satisfy two wave selection constraints from distinct simply connected intervals $\mathcal{T}_\lambda$. However,  the ability to satisfy higher numbers of constraints simultaneously is unclear, potentially limiting the number of distinct simply connected regions where a single WKB solution has support (though of course the full solution to \cref{linear_RDS} will involve a superposition of different modes on different regions $\mathcal{T}_\lambda$). One final important result} is  the  monotonicity of $\mathcal{T}_\lambda$ with respect to $\lambda$, which follows in exactly the same way as presented in \cite{krause_WKB}, which we repeat here for completeness:

\begin{proposition} \label{c1} 
 If $\mathcal{T}_{\lambda_2}\neq\emptyset$ and $0\leq \lambda_1\leq\lambda_2$ then $\mathcal{T}_{\lambda_2}\subseteq \mathcal{T}_{\lambda_1}$. If $\mathcal{T}_{\lambda_1}\neq [0,1]$, and ${0\leq}\lambda_1<\lambda_2$, then we have the stricter inclusion $\mathcal{T}_{\lambda_2}\subset \mathcal{T}_{\lambda_1}$.
\end{proposition}
\begin{proof} 
The first part of this result, for $0\leq \lambda_1\leq\lambda_2$, follows from \cref{newp}. Hence we need to show that if $\lambda_1 < \lambda_2$, then $\mathcal{T}_{\lambda_2}$ is a strict subset of $  \mathcal{T}_{\lambda_1}$ {if $\mathcal{T}_{\lambda_1}\neq [0,1].$} We note that   the internal edges of $\mathcal{T}_\lambda$   are zeros with respect to  $x$ of $$(\tr{\bf B}_\lambda(x))^2-4\det({\bf B}_\lambda(x)).$$ Also, differentiating  with respect to $\lambda$ for $\lambda \geq 0$ gives 
\begin{equation}\label{nmon}
    \frac{\partial}{\partial \lambda} \left[(\tr{\bf B}_\lambda(x))^2-4\det({\bf B}_\lambda(x))\right] = 
    { -2 \tr({\bf B}_\lambda) \tr( \D ^{-1})+4 \tr( \J _\lambda) \det ( \D ^{-1})<0,}
\end{equation}
which follows by determining the sign of each term (all traces and determinants are positive except for $\tr( \J _\lambda)< \tr( \J) <0$). 
As $(\tr{\bf B}_\lambda(x))^2-4\det({\bf B}_\lambda(x)) > 0$ for $x$ in the interior of $\mathcal{T}_\lambda$,  reducing $\lambda_2$
increases the value of 
\begin{equation}\label{b0l2} \tr({\bf B}_{\lambda_2}(x))^2-4\det({\bf B}_{\lambda_2}(x))
\end{equation} 
at any given point $x$. Hence if an open simply connected region within $\mathcal{T}_{\lambda_2}$ is given by 
$a(\lambda_2) < x < b(\lambda_2)$, we have that if $0\leq\lambda_1<\lambda_2$ and { $a(\lambda_2) > 0$, so that the edge point is in the interior of the domain and thus a zero of Eq.~\eqref{b0l2},}  then $a(\lambda_1) < a(\lambda_2)$. Similarly, once more given  the above ordering of $\lambda_1, \, \lambda_2$, if $b(\lambda_2) < 1$ then $b(\lambda_1) > b(\lambda_2)$, so the strict inclusion $\mathcal{T}_{\lambda_2}\subset \mathcal{T}_{\lambda_1}$ for $\mathcal{T}_{\lambda_1}\neq [0,1]$ follows.
\end{proof}
 
 \subsection{\cref{het_prop}:}

 We are now in a position to prove \cref{het_prop}:

\begin{proof}
By \cref{p3}, we can specialize to non-negative growth rates, $\lambda$ that are real. We consider $\mathcal{T}_0$, which by the  conditions, \cref{hcds}, \eqref{het_instab_cond} given in the theorem is non-empty, and by the inclusion result, \cref{c1}, contains $\mathcal{T}_\lambda$ for any non-negative $\lambda$. Continuity ensures that there exists sufficiently small $\lambda>0$ such that $\mathcal{T}_\lambda$ is non-empty and the simple zeros of $\mathcal{T}_0$ continuously map to simple zeros of $\mathcal{T}_\lambda$ as $\lambda$ is continuously shifted away from zero. 

If we further have the limitation that $$[\tr( \Bl (x))]^2-4 \det( \Bl (x))$$ has no more than two simple  zeros for $x \in [0,1]$  with $\mathcal{T}_\lambda$
simply connected for any such $\lambda$ then \cref{heteroturing} gives  leading order (in $\ep$)   WKB solutions, $\bm w_\pm$ of \cref{eq:PiecewiseWKBsolu_full}, that satisfy the requirements of   \cref{het_prop}. In particular, they are  non-trivial, inhomogeneous, bounded and will drive an instability on $\mathcal{T}_\lambda \subseteq \mathcal{T}_0$ providing the wave selection criterion, as given by 
\cref{hetselect}, can be satisfied. In particular, noting \cref{eps} and \cref{hetselect}, continuously reducing the diffusion scale $D$ to be sufficiently small and thus continuously reducing $\ep$,  will locate  a value of $D$ for which the wave selection criterion can be satisfied. 

For the more general case where $$[\tr( \Bl (x))]^2-4 \det( \Bl (x))$$ can have more than two simple  zeros for $x \in [0,1]$ and/or  $\mathcal{T}_\lambda$ is not 
necessarily simply connected, we can similarly construct leading order   WKB solutions  of the form of $\bm w_\pm$ in \cref{eq:PiecewiseWKBsolu_full}, with a  wave selection constraint of the form of \cref{hetselect} for  each simply connected region making up $\mathcal{T}_\lambda.$ One, of many possible, wave selection criteria for this solution then is the constraint of satisfying the wave selection criterion on one, {and only one,} of the simply connected regions making up $\mathcal{T}_\lambda$. As this is, once more,  a single simply connected region it  can be achieved by  
continuously reducing the diffusion scale $\ep$ to be sufficiently small, as above. 
An associated leading order WKB solution can be constructed by taking it to be trivial on the other simply connected regions constituting $\mathcal{T}_\lambda$ and of the form of 
$\bm w_\pm$ in  \cref{eq:PiecewiseWKBsolu_full}
 for  the simply connected region where the wave selection criterion is satisfied. This provides  leading order WKB  solutions with the appropriate properties of being  non-trivial, inhomogeneous and bounded, as   required to demonstrate the theorem. 
\end{proof}

Note that the conditions of   \cref{het_prop}, with the homogeneous stability condition of \cref{hcds} for $x\in[0,1]$ 
and the inhomogeneous condition of \cref{het_instab_cond} on at least a subset of  $x\in[0,1]$, directly generalize both the homogeneous conditions of \cref{hom_prop} and the {reaction-diffusion} system without cross-diffusion studied in \cite{krause_WKB}. Thus Turing instability conditions generalize to inhomogeneous cross-diffusion systems, of which chemotaxis is a special case, with the very weak additional requirement that any zeros of 
$[\tr( \Bz (x))]^2-4 \det( \Bz (x))$ are simple, noting  zeros are excluded from the homogeneous case by \cref{het_instab_cond}.

Furthermore, 
by inspection of \cref{hetselect},
there are countably infinite values of the diffusion scale, $D$, or equivalently $\ep$, where the wave selection criterion is satisfied. One may also  trivially note from the proof of \cref{het_prop} that when $\mathcal{T}_0$ is more complex, in particular not simply connected,  there is a collection of WKB solutions of the form of  \cref{eq:PiecewiseWKBsolu_full} 
 for  each   simply connected region constituting   $\mathcal{T}_\lambda$ and trivial elsewhere, generating multiple leading order WKB solutions. 
Once we have that regions given by the zeros of $[\tr( \Bz (x))]^2-4 \det( \Bz (x))$ are not simply connected, delimiting  when the wave selection criteria on more than one simply connected constituent of $\mathcal{T}_\lambda$ can be {\it simultaneously} satisfied requires quantitative knowledge of the eigenvalues $\mu_\lambda^\pm$, which we have not considered. {However, we still anticipate being able to satisfy wave selection criteria on two simply connected constituents of $\mathcal{T}_\lambda$ as there are two degrees of freedom, namely $\lambda$ and, for example via changes in the diffusion scale, $\ep$, that can be adjusted. More generally, this demonstrates the  prospect} of an even  more extensive  class of leading order WKB solutions,  which have support on more than one simply connected domain and  highlights rich possibilities of unstable leading order WKB solutions for more complex $\mathcal{T}_\lambda$. 
Such possibilities, as well as a more general study of these systems and the instability conditions, will be numerically explored in the following sections, where we explicitly 
demonstrate that heterogeneity can serve to localize patterns to regions predicted by the simple inequalities in \cref{het_prop}.

\bk 
\section{Numerical Results}\label{Numerical_Sect}
We {have}  built a general MATLAB code to solve \eqref{orig_eqn} numerically, as well as to compute the boundaries of $\mathcal{T}_0$ where we expect to find deviations from our approximate steady state given by $\f(\us,x)=\bm{0}$. The code implements a method-of-lines discretization of the cross-diffusion terms given by
\begin{equation}
    \pd{}{x}\left (D(u,v,x)\pd{u}{x} \right ) \approx \frac{(u_{i+1}-u_i)(D_{i+1}+D_i)-(u_i-u_{i-1})(D_i+D_{i-1})}{2\alk{\delta} x^2},
\end{equation}
where $D_i = D(u_i,v_i,x_i)$ with the subscript indicating evaluation at the $i$th grid point, and $\alk{\delta} x = 1/(N-1)$ with $N$ denoting the number of grid points. The resulting system of ODEs is integrated using the MATLAB function \emph{ode15s}, which implements a variable-step, variable-order solver \cite{Shampine1997}. Relative and absolute tolerances are taken to be $10^{-11}$ and a Jacobian sparsity pattern was used to speed up timestepping. Unless otherwise noted, $N=10^4$ equispaced grid points are used. Convergence checks in the number of grid points and the maximum timestep were used for selected simulations to ensure soundness of the method. In all cases we used initial data of the form $u(0,x) = u^*(x)\xi_u(x)$, $v(0,x) = v^*(x)\xi_v(x)$, with $\xi_u,\xi_v$ being normal random variables with mean 1 and standard deviation $0.1$ independently and identically at each spatial point. The code can be found at \cite{krause_git}, and we encourage the interested reader to explore other systems and parameter regimes than what we report here.

To demonstrate the generality of our theory, we study three examples of cross-diffusion and its impact on pattern formation. Firstly we consider the Schnakenberg model \cite{schnakenberg1979simple,murray2004mathematical} with linear cross-diffusion terms \cite{gambino2016super},
\beq
\begin{aligned}\label{Sch_eqns}
    \pd{u}{t} =& \ep^2 \pd{}{x}\left (D_{11}(x)\pd{u}{x}+D_{12}(x)\pd{v}{x} \right ) + a(x)-u+u^2v,\\
    \pd{v}{t} =& \ep^2 \pd{}{x}\left (D_{21}(x)\pd{u}{x}+D_{22}(x)\pd{v}{x} \right ) + b(x)-u^2v,
    \end{aligned}
\eeq
where we assume $a(x)>0, b(x)>0$ for all $x \in [0,1]$.  The approximate steady state is given by $u^*(x) =a(x)+b(x),$ $v^*(x) = b(x)/(a(x)+b(x))^2$. 

We also consider a version of the Keller-Segel model of chemotaxis \cite{keller1971model,horstmann20031970,murray2004mathematical,hillen2009user}, noting that spatially heterogeneous variants have been studied in \cite{yan2021keller} and references therein. Accounting for logistic cell growth and linear chemoattractant dynamics, the model is given by
\beq
\begin{aligned}\label{KS_eqns}
    \pd{u}{t} =& \ep^2 \pd{}{x}\left (D_{11}(x)\pd{u}{x}-\chi(x) u\pd{v}{x} \right ) + u\left(1-\frac{u}{K(x)}\right),\\
    \pd{v}{t} =& \ep^2 \pd{}{x}\left (D_{22}(x)\pd{v}{x} \right ) + h(x)u-v,
    \end{aligned}
\eeq
where we assume $K(x)>0, h(x)>0$ and $\chi(x)>0$ for all $x \in [0,1]$. The approximate steady state is given by $u^*(x) = K(x)$, $v^*(x) = h(x)K(x)$.

Finally we explore the classical Shigesada-Kawasaki-Teramoto (SKT) model \cite{shigesada1979spatial,okubo2001diffusion, le2005regularity,choi2004existence,ruiz2013mathematical}, with spatially heterogeneous variations studied in \cite{kuto2009stability} and elsewhere.  After a suitable rescaling, this model can be written\footnote{Starting from the original SKT model with a heterogeneous diffusion tensor, one may also obtain an `advection' term involving first-order derivatives in $u$ and $v$. For simplicity we will neglect these terms, as they do not arise in almost all other work on the SKT model when the diffusion tensor does not depend on $x$.} in a cross-diffusion form given by
\beq
\begin{aligned}\label{SKT_eqns}
    \pd{u}{t} =& \ep^2 \pd{}{x}\left ((d_1(x)+d_{11}(x)u+d_{12}(x)v)\pd{u}{x}+d_{12}(x)u\pd{v}{x} \right ) + r_1(x)u(1-a_1(x)u-b_1(x)v),\\
    \pd{v}{t} =& \ep^2 \pd{}{x}\left (d_{21}(x)v\pd{u}{x}+(d_2(x)+d_{21}(x)u+d_{22}(x)v)\pd{v}{x} \right ) + r_2(x)v(1-b_2(x)u-a_2(x)v),
    \end{aligned}
\eeq
    where we assume that the kinetic functions $r_i(x)>0, a_i(x)>0$, and $b_i(x)>0$ for all $x\in [0,1]$. We will also assume that $a_1(x)>b_2(x)$, and $a_2(x)>b_1(x)$ to ensure the feasibility and stability (in the absence of transport) of an approximate coexistence equilibrium given by $u^*(x) = (a_2(x)-b_1(x))/(a_1(x)a_2(x)-b_1(x)b_2(x))$, $v^*(x) = (a_1(x)-b_2(x))/(a_1(x)a_2(x)-b_1(x)b_2(x))$. 
    
    In all cases we assume that the kinetic parameters are chosen so that our approximate heterogeneous steady states, $\us$, satisfy the boundary conditions \eqref{BCs}. This is not strictly needed for the parameters $r_i$ in the SKT model, as these do not enter into our approximate steady states. We do not make any explicit restrictions on the cross-diffusion parameters, except that $\D$ must remain positive-definite. We note, however, that the existence and regularity theory for these systems is much more intricate than for simpler reaction-diffusion models, with blowup and singularities having a significant literature \cite{choi2004existence,le2005regularity, kuto2009stability, yan2021keller}; see \cite{lankeit2020facing} for an introductory review to these complexities and their analysis. 
    
\begin{figure}
    \centering
    \subfloat[$\ep=0.006$]{\includegraphics[width=0.49\textwidth]{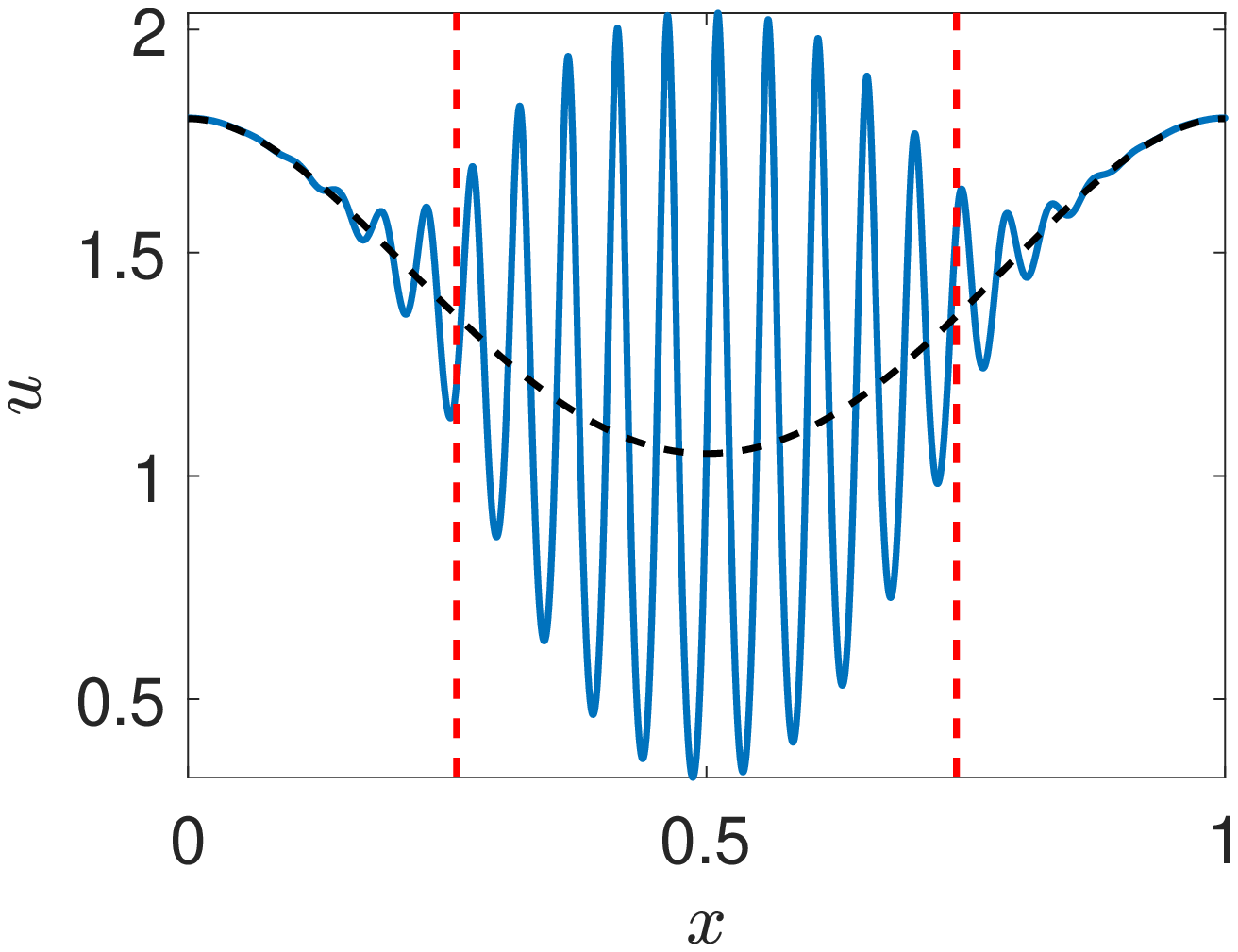}}
    \hspace{0.2cm}\subfloat[$\ep=0.003$]{\includegraphics[width=0.49\textwidth]{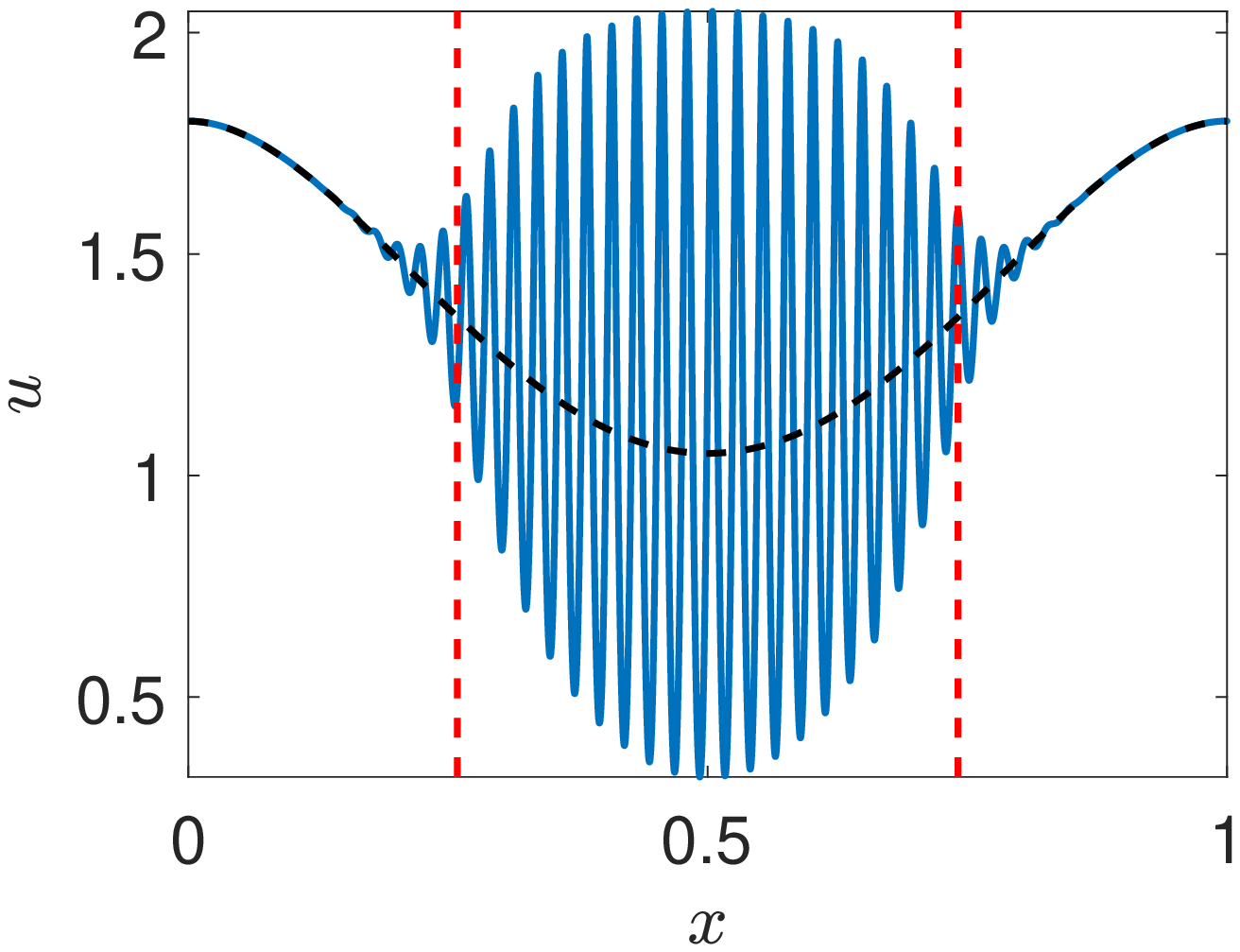}}
        
    \subfloat[$\ep=0.001$]{\includegraphics[width=0.49\textwidth]{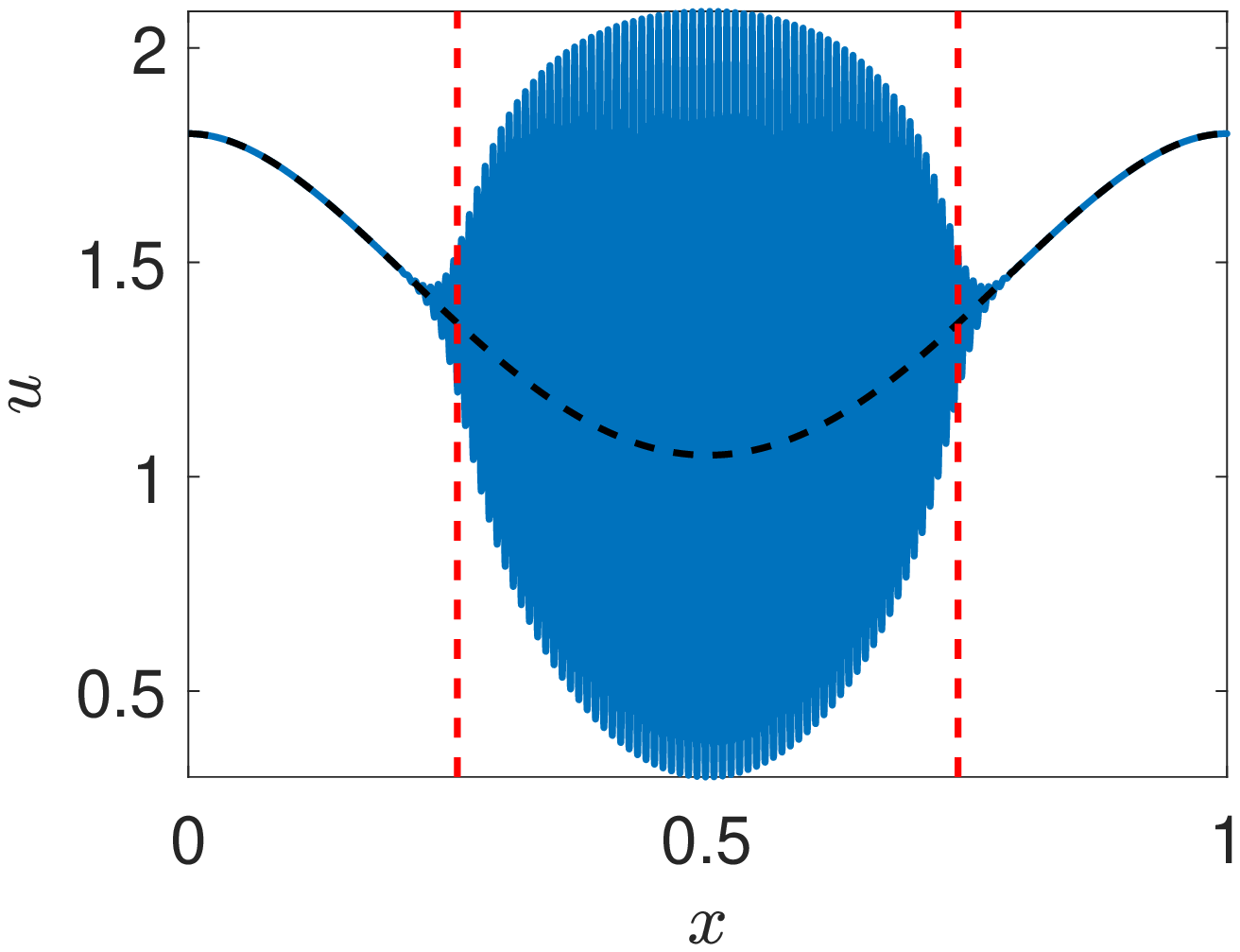}}
    \hspace{0.2cm}\subfloat[$\ep=0.0003$]{\includegraphics[width=0.49\textwidth]{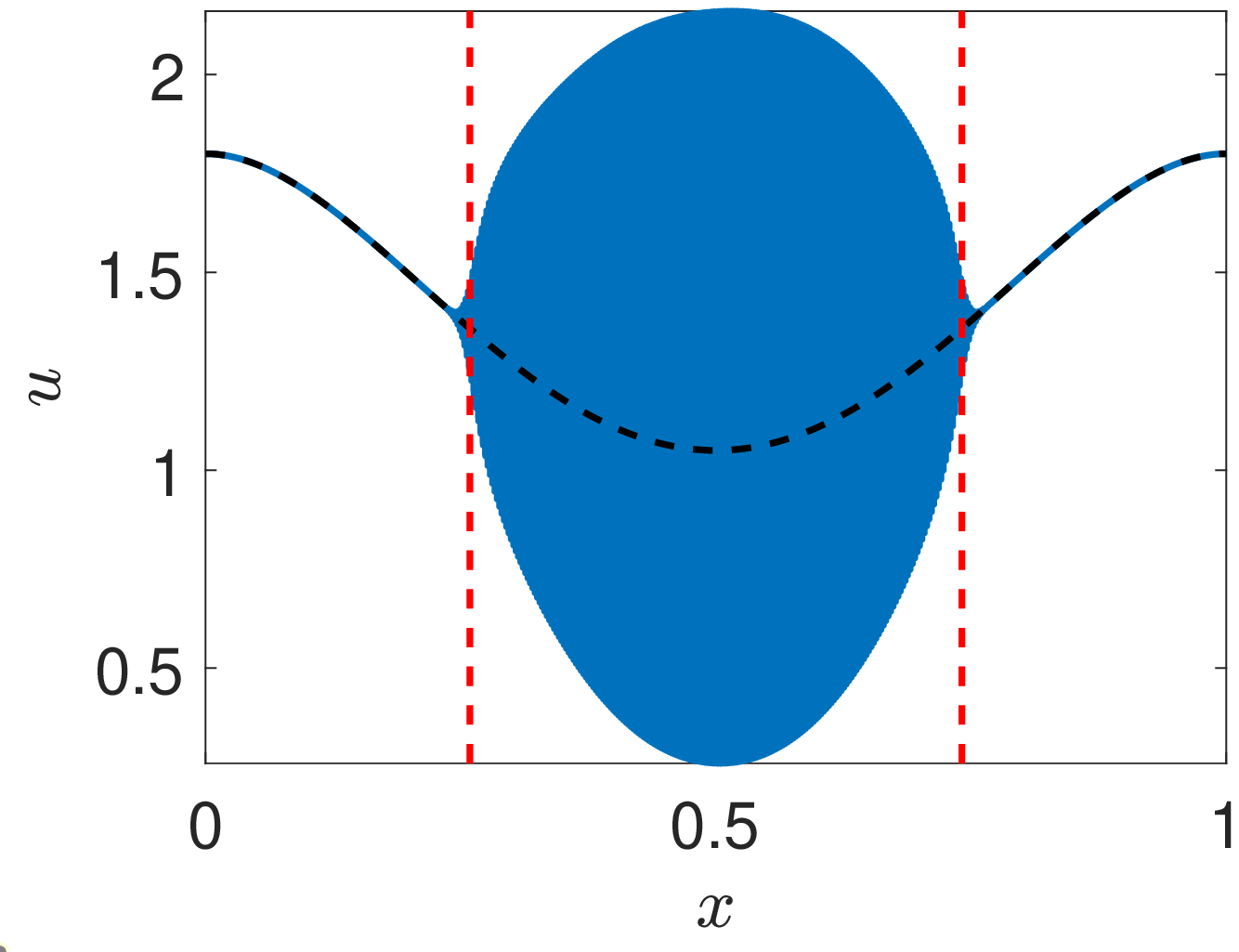}}
    \caption{Plots of $u$ in blue curves and $u^*$ in dashed black curves from solutions of {the Schnakenberg model, \cref{Sch_eqns},}  for various values of $\ep$ at $T=5,000$, which are essentially at steady state. The red  {vertical }  lines show the boundary of $\mathcal{T}_0$ computed from the conditions in \cref{het_prop}. The parameters are taken as $a(x) = 0.8-12x^2(x-1)^2$, $b=1$, $D_{11}=D_{22}=D_{12}=1$, $D_{21}=3((x-0.5)^2-1)$. Note that in panel (D), $N=5\times 10^4$ grid points were used to accurately represent the solution.}
    \label{SchStat}
\end{figure}

We start by exploring {the Schnakenberg model, \cref{Sch_eqns},}  in a regime where the pure reaction-diffusion system would not pattern, namely when $D_{11}=D_{22}=1$. We make $a(x)$ and $D_{21}(x)$ depend on space in such a way to localize patterns to the interior of the domain, noting that $a(x)$ must satisfy the boundary conditions but $D_{21}(x)$ need not. We plot our numerical simulations in \cref{SchStat} for decreasing values of $\ep$. The solution is shown at time $t=5,000$, but in all cases is indistinguishable from the solution at $t=200$, hence this appears to be an approximate (patterned) steady state. As anticipated by the theory, the location of the Turing regime $\mathcal{T}_0$ approximates where deviations from the heterogeneous steady state occur, and this approximation becomes better as $\ep$ decreases, with the number of internal `spikes' in the pattern increasing. This localization was shown for relatively simple heterogeneities in \cite{krause_WKB}, and so we will now consider more elaborate examples.
    
\begin{figure}
    \centering
    \subfloat[$\ep=0.006$]{\includegraphics[width=0.49\textwidth]{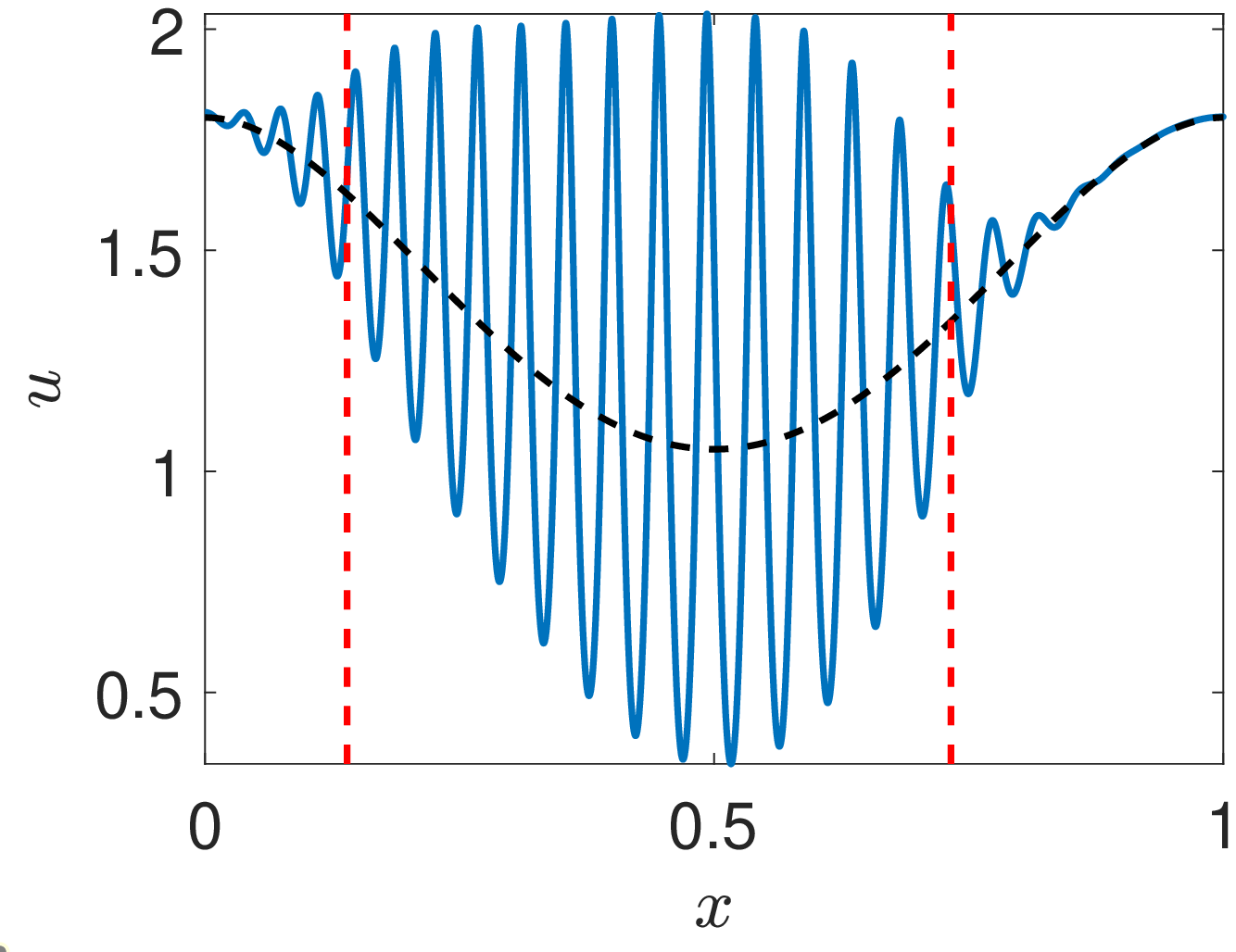}}\hspace{0.2cm}\subfloat[$\ep=0.006$]{\includegraphics[width=0.49\textwidth]{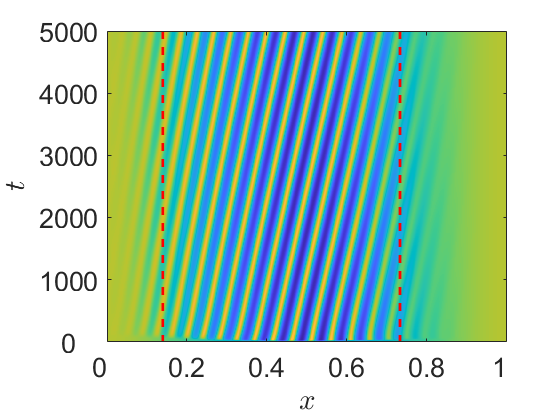}}
        
    \subfloat[$\ep=0.003$]{\includegraphics[width=0.49\textwidth]{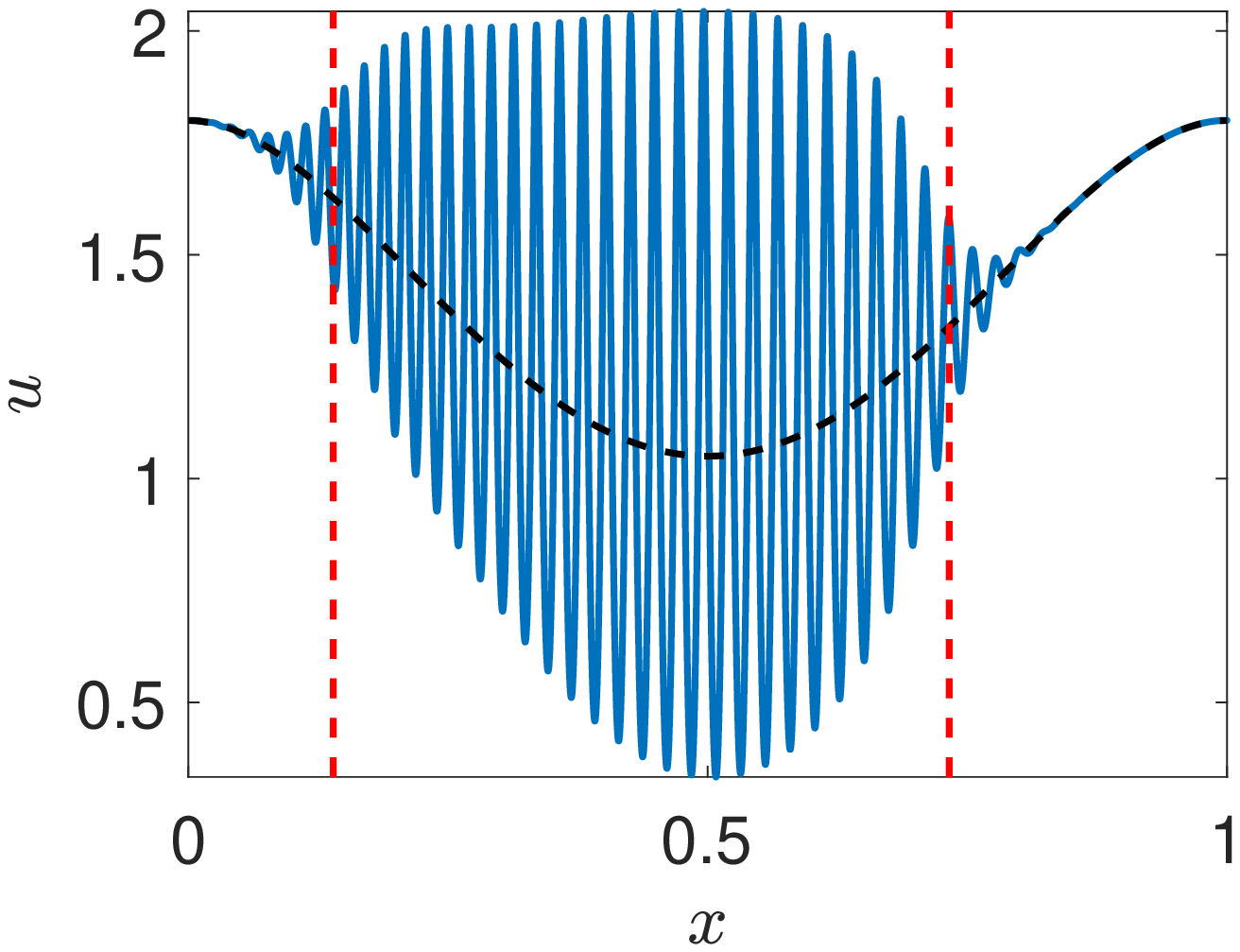}}\hspace{0.2cm}\subfloat[$\ep=0.003$]{\includegraphics[width=0.49\textwidth]{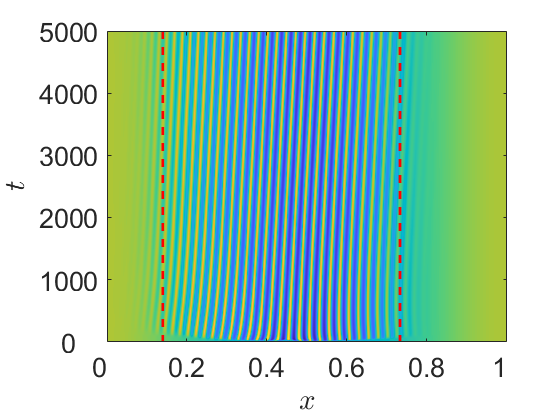}}
                
    \subfloat[$\ep=0.001$]{\includegraphics[width=0.49\textwidth]{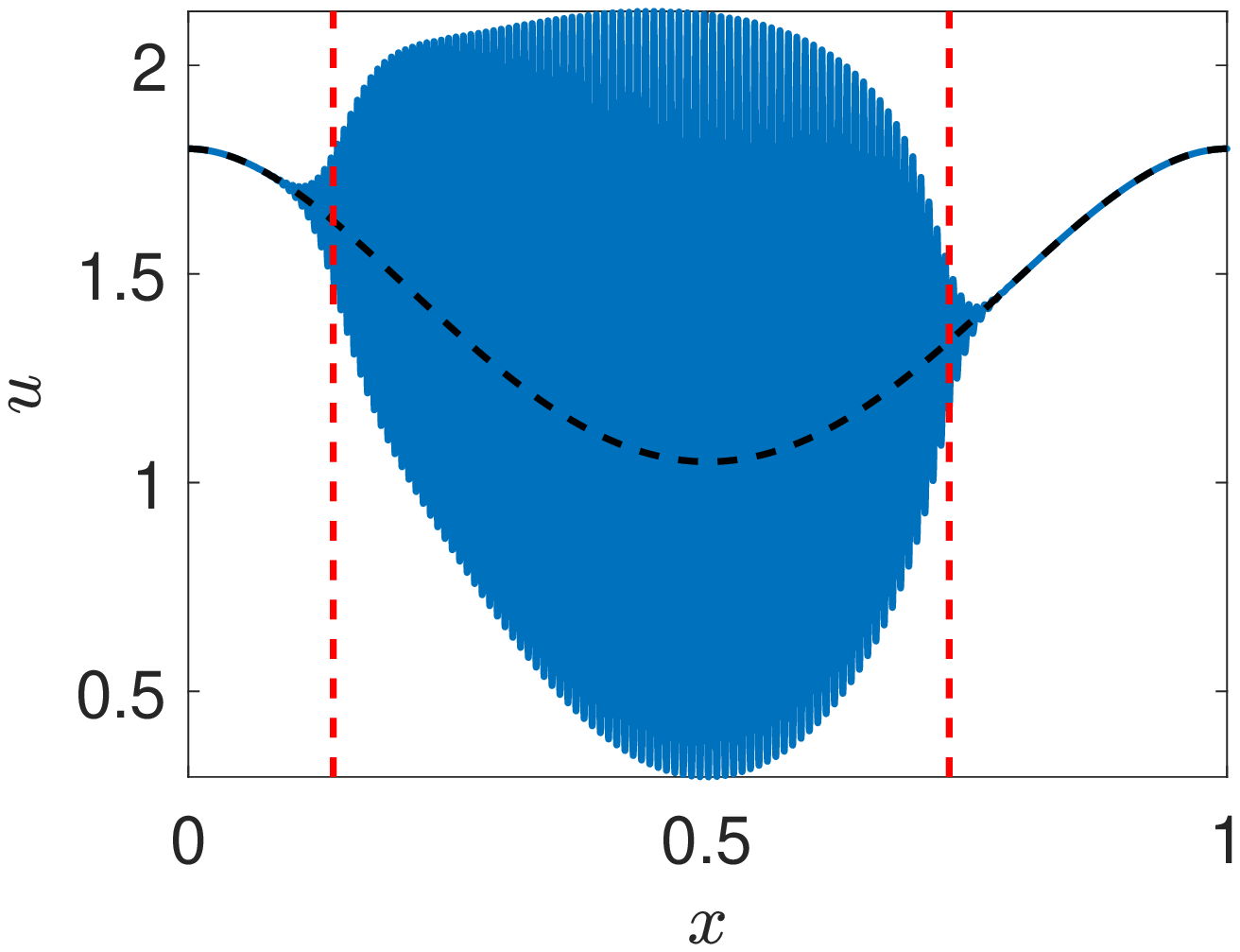}}\hspace{0.2cm}\subfloat[$\ep=0.001$]{\includegraphics[width=0.49\textwidth]{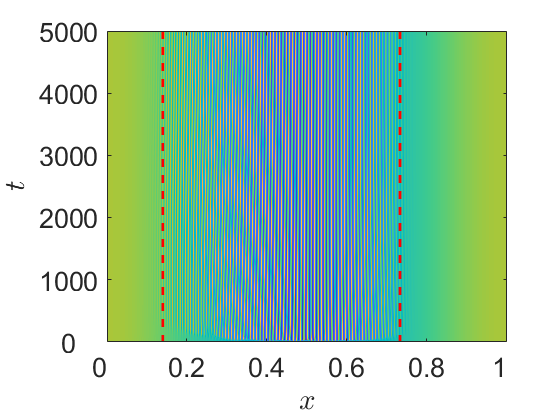}}
    \caption{Plots of $u$ in blue curves and $u^*$ in dashed black curves from solutions {of the Schnakenberg model, \cref{Sch_eqns},}  for various values of $\ep$ at $T=5,000$ in (A), \alk{(C)}, (E), and kymographs of $u$ in (B), (D), (F). The red  {vertical }  lines show the boundary of $\mathcal{T}_0$ computed from the conditions in \cref{het_prop}. The parameters are taken as $a(x) = 0.8-12x^2(x-1)^2$, $b=1$, $D_{11}=D_{22}=1$, $D_{12}=0.5+0.8x$,  $D_{21}=3((x-0.5)^2-1)$. }
    \label{SchStatMov}
\end{figure}

The heterogeneities used in \cref{SchStat} were symmetric about the midpoint of the domain $x=0.5$. We next consider an example with an asymmetric heterogeneity in $D_{12}(x)$, but otherwise use the same parameters. We show solutions for decreasing $\ep$ in \cref{SchStatMov}, and now include kymographs or space-time plots showing the evolution of $u$ over time. Unlike the previous example, the solutions do not reach an apparent steady state but now move in the direction of increasing $D_{12}(x)$. The speed of these moving spikes is influenced by $\ep$ with extremely slow movement seen in panels (D) and especially (F). We expect that this movement is due to the kind of heterogeneity-induced spike oscillations reported in \cite{page2005complex}, which were later studied and explained in terms of spike generation and annihilation in \cite{kolokolnikov2018pattern, krause2018heterogeneity}. Despite the spatiotemporal nature of the solutions, the boundaries of $\mathcal{T}_0$ still give a good approximation for where the deviations from the steady state $u^*$ occur.
    
\begin{figure}
    \centering
    
    \subfloat[$\ep=0.005$]{\includegraphics[width=0.44\textwidth]{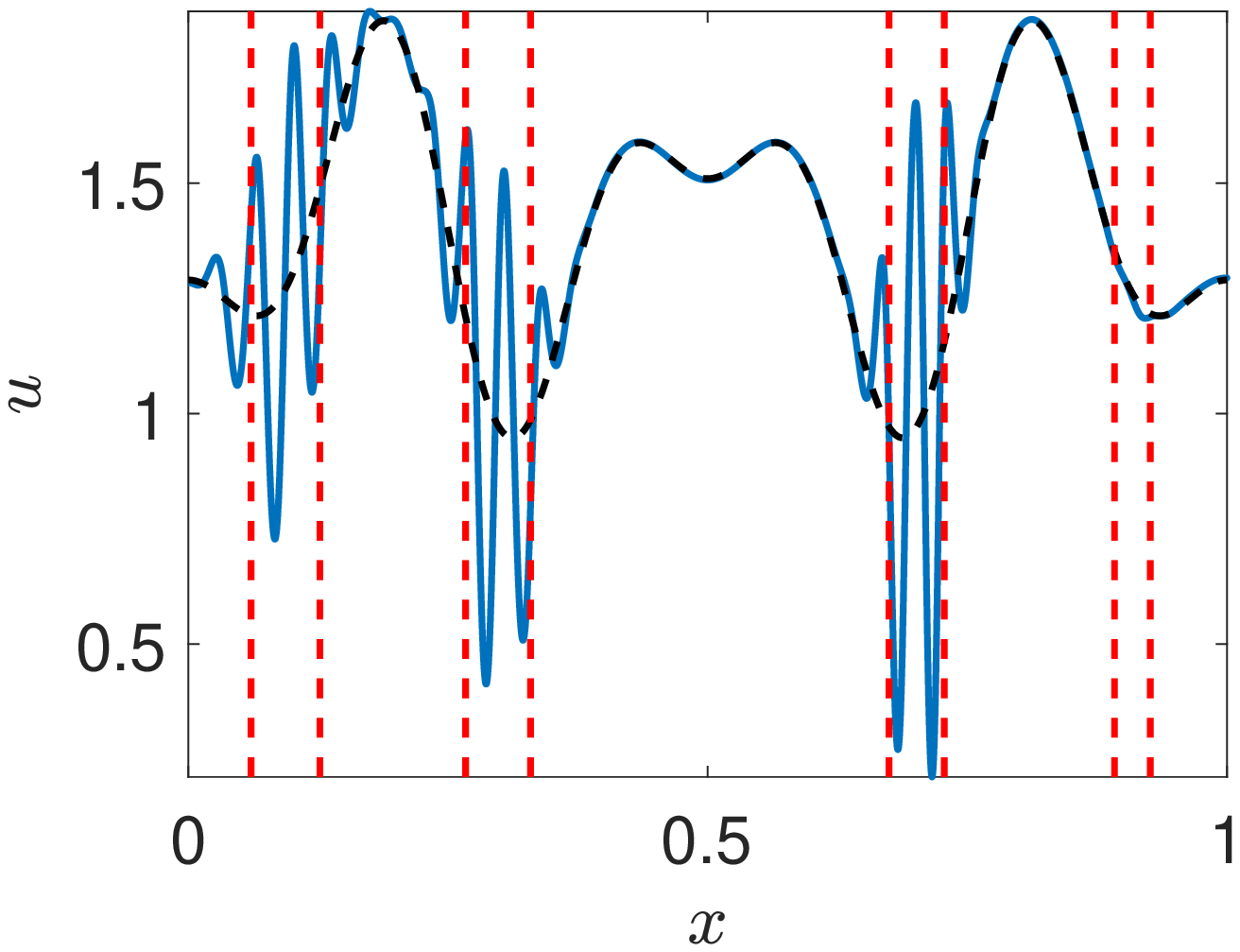}}\hspace{0.2cm}\subfloat[$\ep=0.005$]{\includegraphics[width=0.44\textwidth]{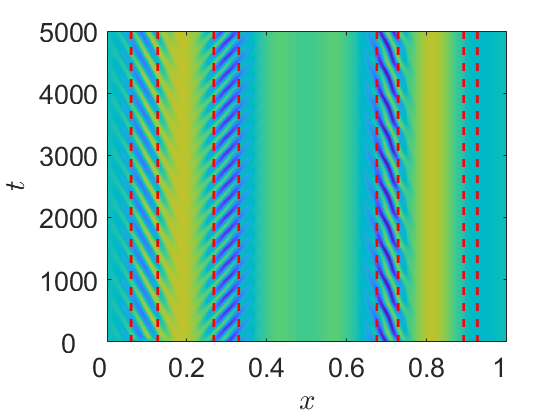}}
    
    \subfloat[$\ep=0.004$]{\includegraphics[width=0.44\textwidth]{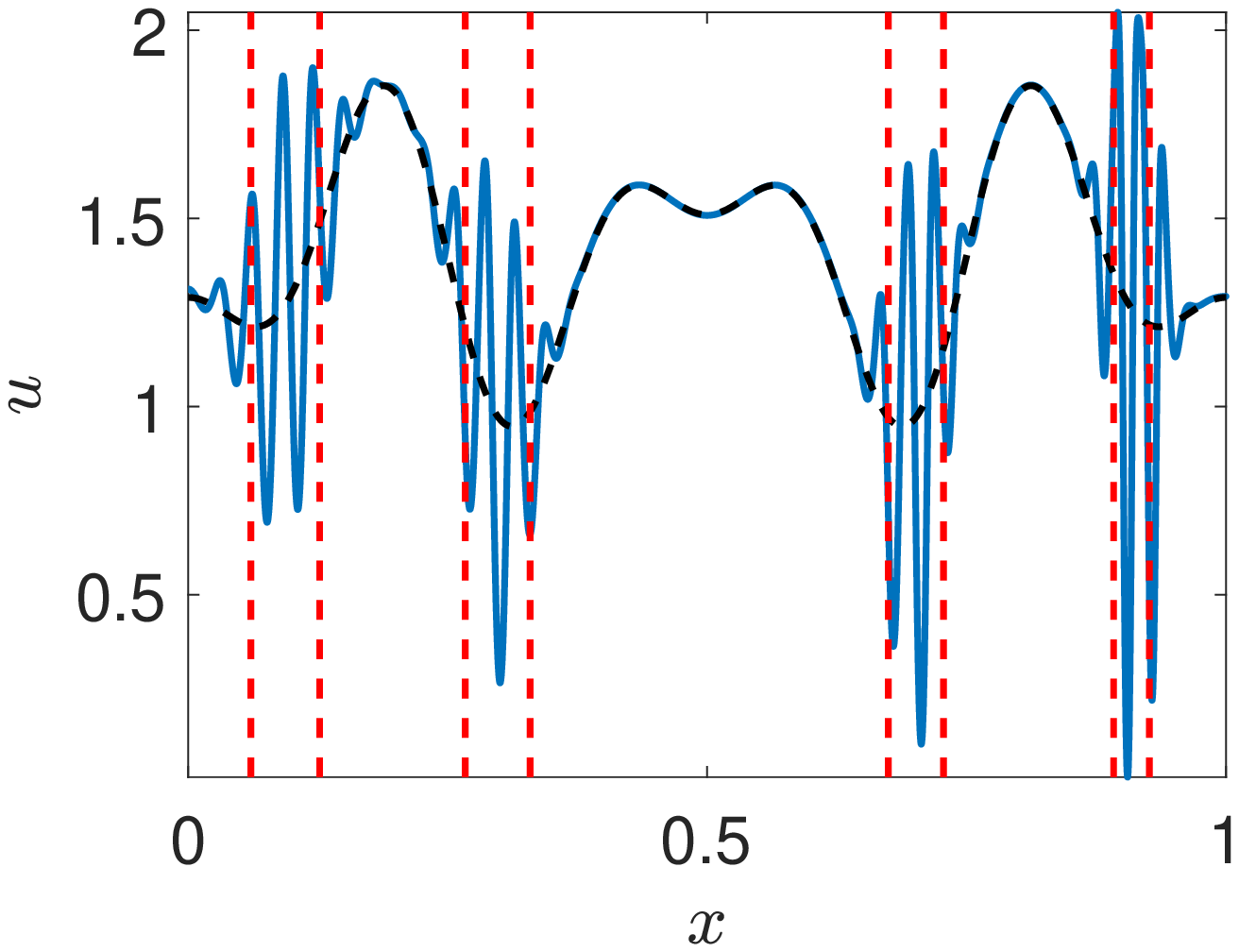}}\hspace{0.2cm}\subfloat[$\ep=0.004$]{\includegraphics[width=0.44\textwidth]{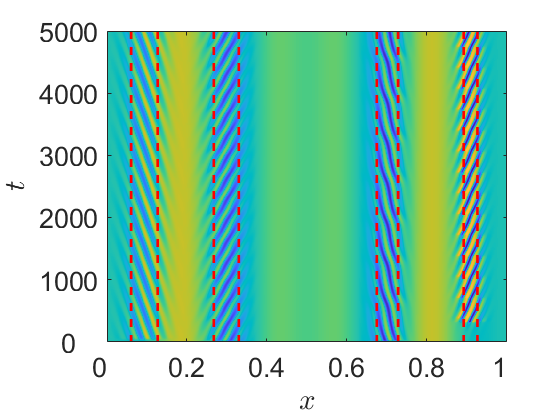}}
    
    \subfloat[$\ep=0.002$]{\includegraphics[width=0.44\textwidth]{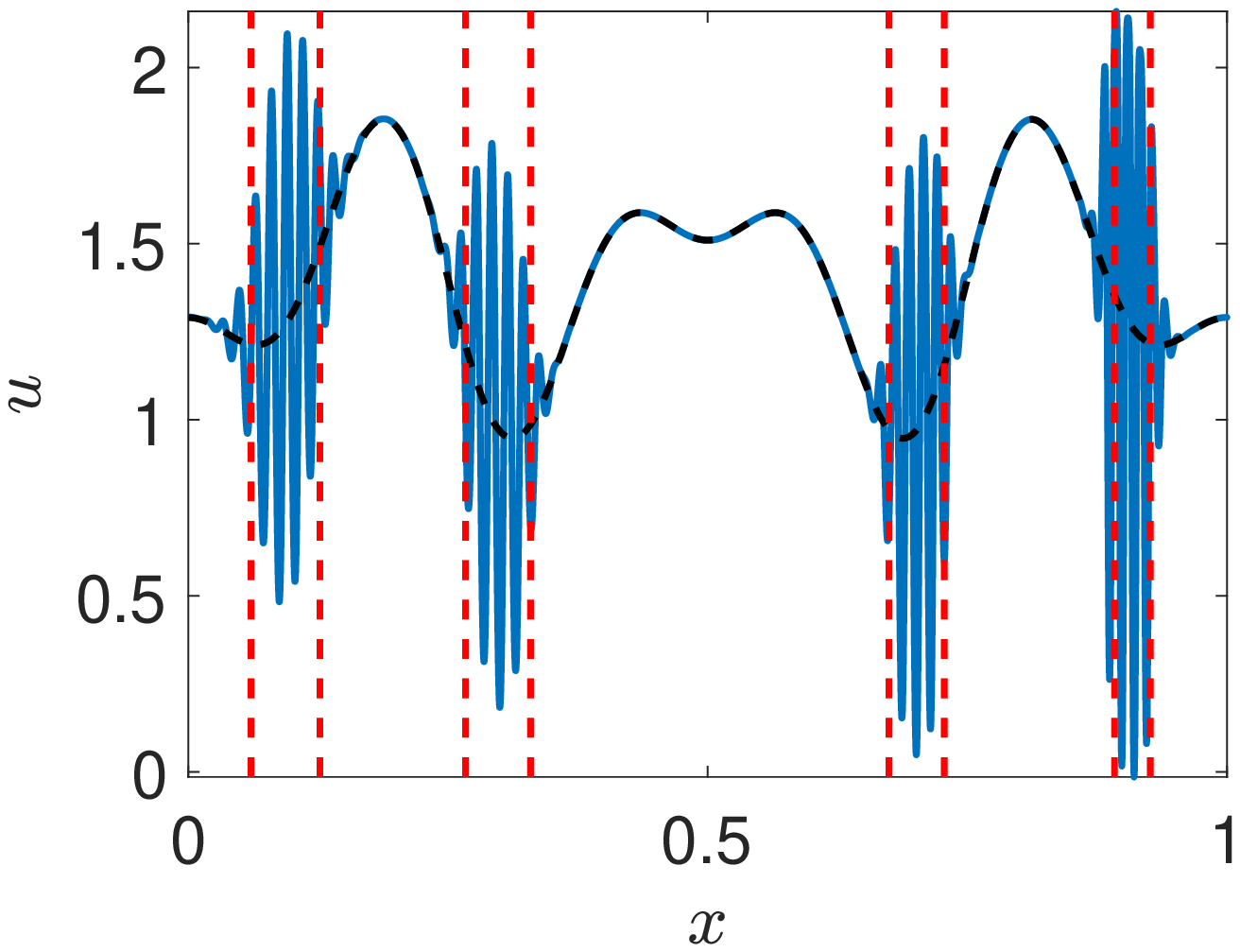}}\hspace{0.2cm}\subfloat[$\ep=0.002$]{\includegraphics[width=0.44\textwidth]{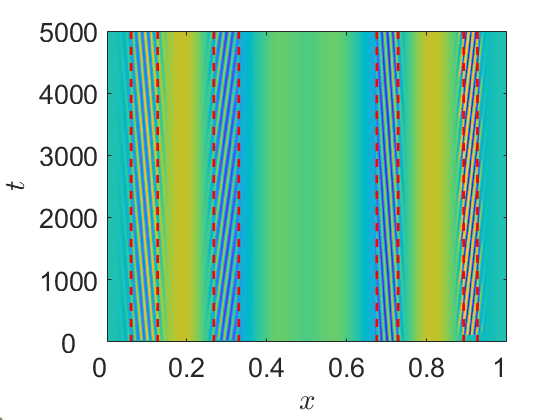}}

    \subfloat[$\ep=0.001$]{\includegraphics[width=0.44\textwidth]{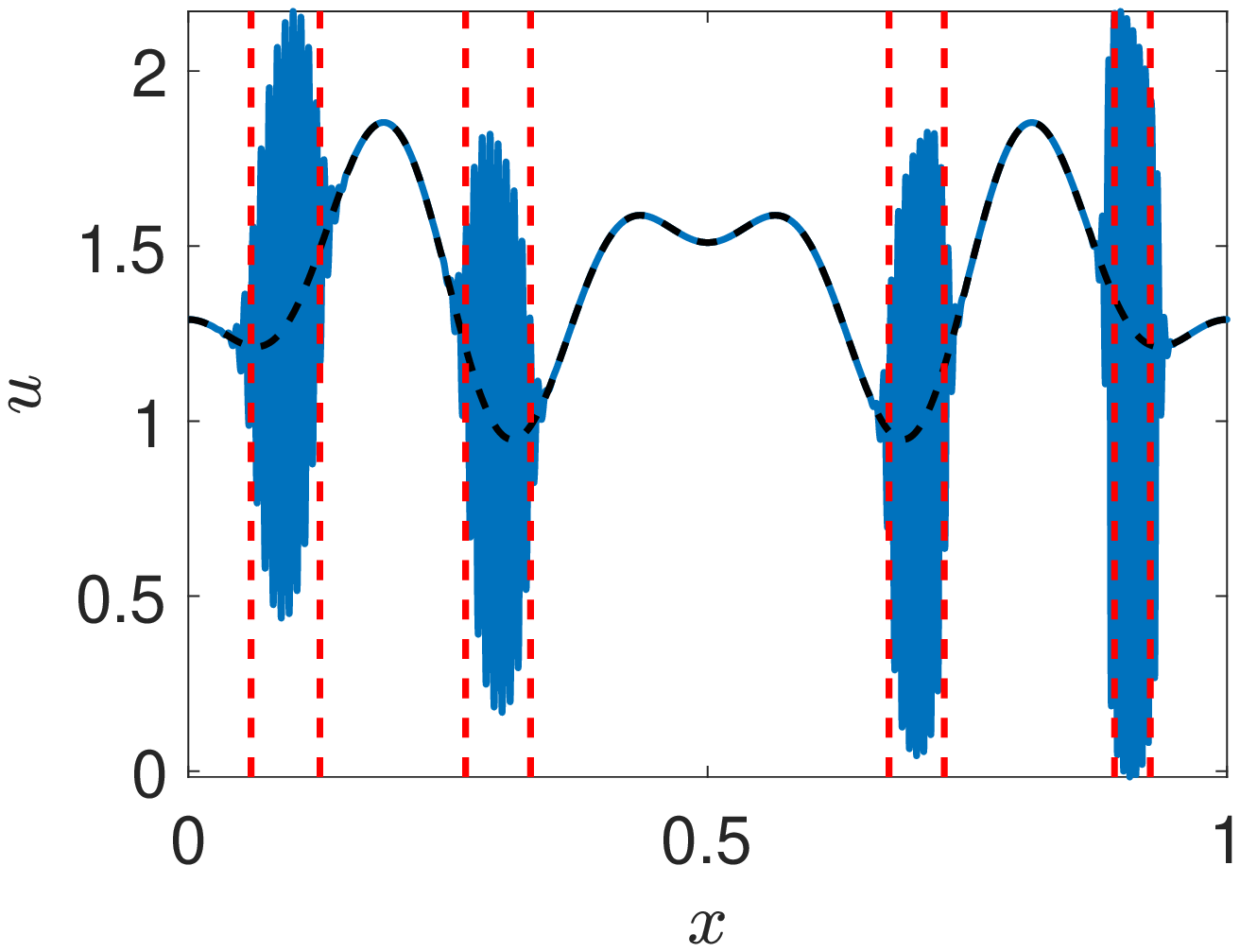}}\hspace{0.2cm}\subfloat[$\ep=0.001$]{\includegraphics[width=0.44\textwidth]{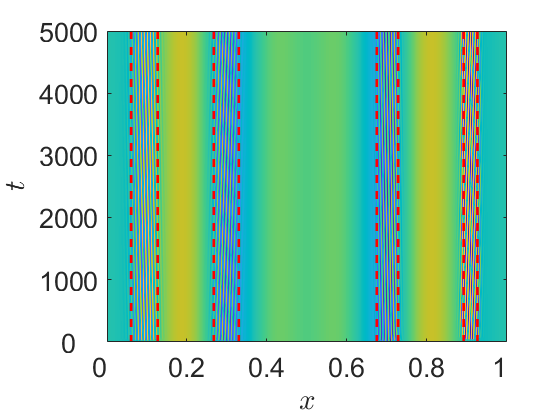}}
    \caption{Plots of $u$ in blue curves and $u^*$ in dashed black curves from solutions of {the Schnakenberg model, \cref{Sch_eqns},}  for various values of $\ep$ at $T=5000$ in (A), \alk{(C)}, (E), and kymographs of $u$ in (B), (D), (F). The red  {vertical }  lines show the boundary of $\mathcal{T}_0$ computed from the conditions in \cref{het_prop}. The parameters are taken as $a(x) = 0.01+0.19(1+\cos(10 x\pi))$, $b=0.9+0.3(1-\cos(6 x\pi))$, $D_{11}=D_{22}=1$, $D_{12}=1+\sin(3x\pi)$,  $D_{21}=2(x-1)$.}
    \label{SchMov}
\end{figure}

Finally we consider a more intricate example with complex heterogeneities in both cross-diffusion parameters, $a(x)$, and $b(x)$ in \cref{SchMov}. There are now four regions where \cref{het_prop} predicts pattern formation. However, for sufficiently large $\ep$, wavemode selection prevents the rightmost of these regions from patterning in panels (A) and (B). As in \cref{SchStatMov}, solutions within each region are spatiotemporally moving spikes, with speeds varying with $\ep$ but also with speed and direction dependent on the local heterogeneity. In particular, the local speed of a spike decreases with increasing $x$, and the direction of movement is always from lower values of $u^*$ to higher values. We also observe an increase in the local frequency of oscillations with increasing $x$ (see especially the leftmost and rightmost patterned regions in panels (C) and (E)).

\begin{figure}
    \centering
    \subfloat[$\ep=0.01$]{\includegraphics[width=0.49\textwidth]{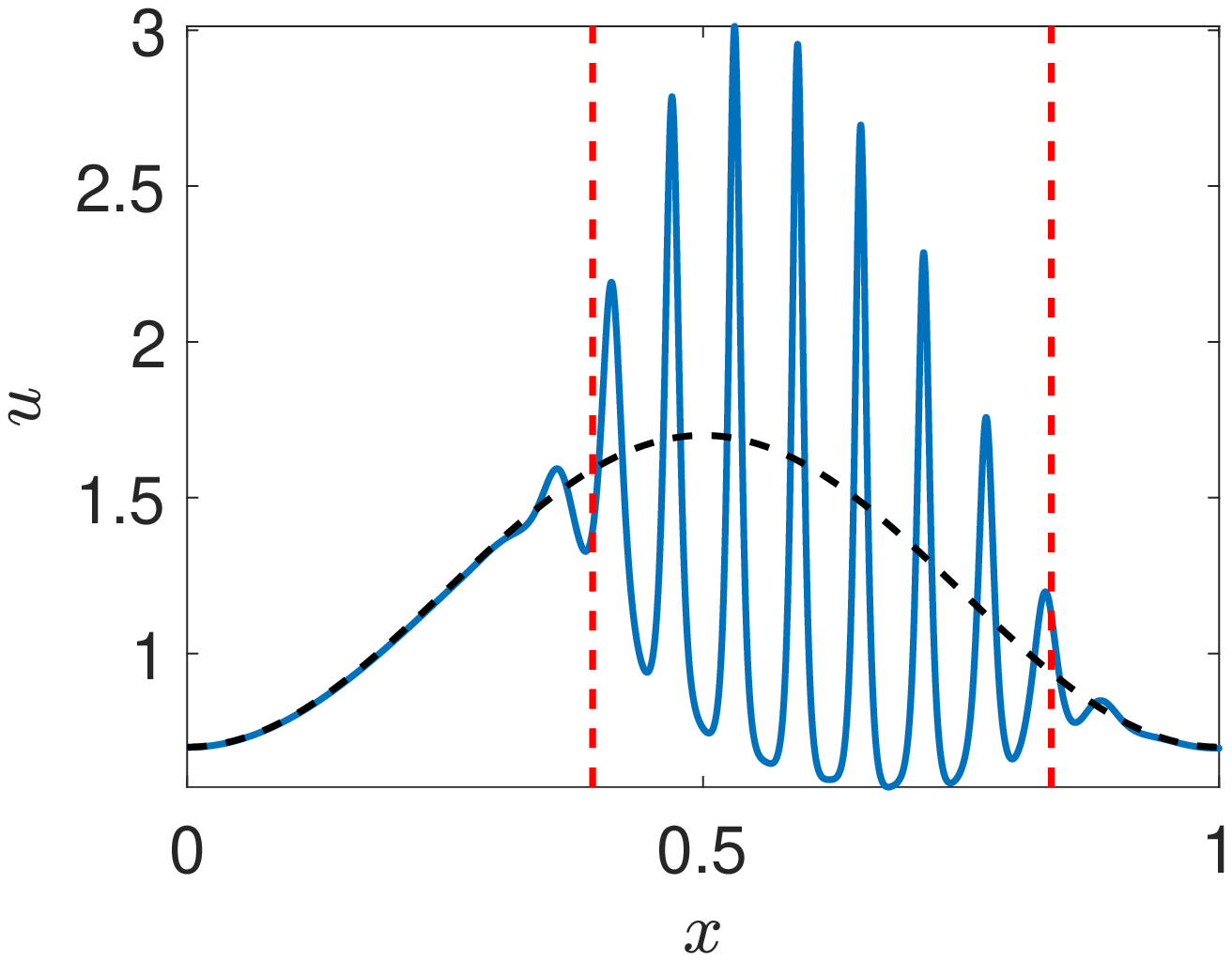}}\hspace{0.2cm}\subfloat[$\ep=0.006$]{\includegraphics[width=0.49\textwidth]{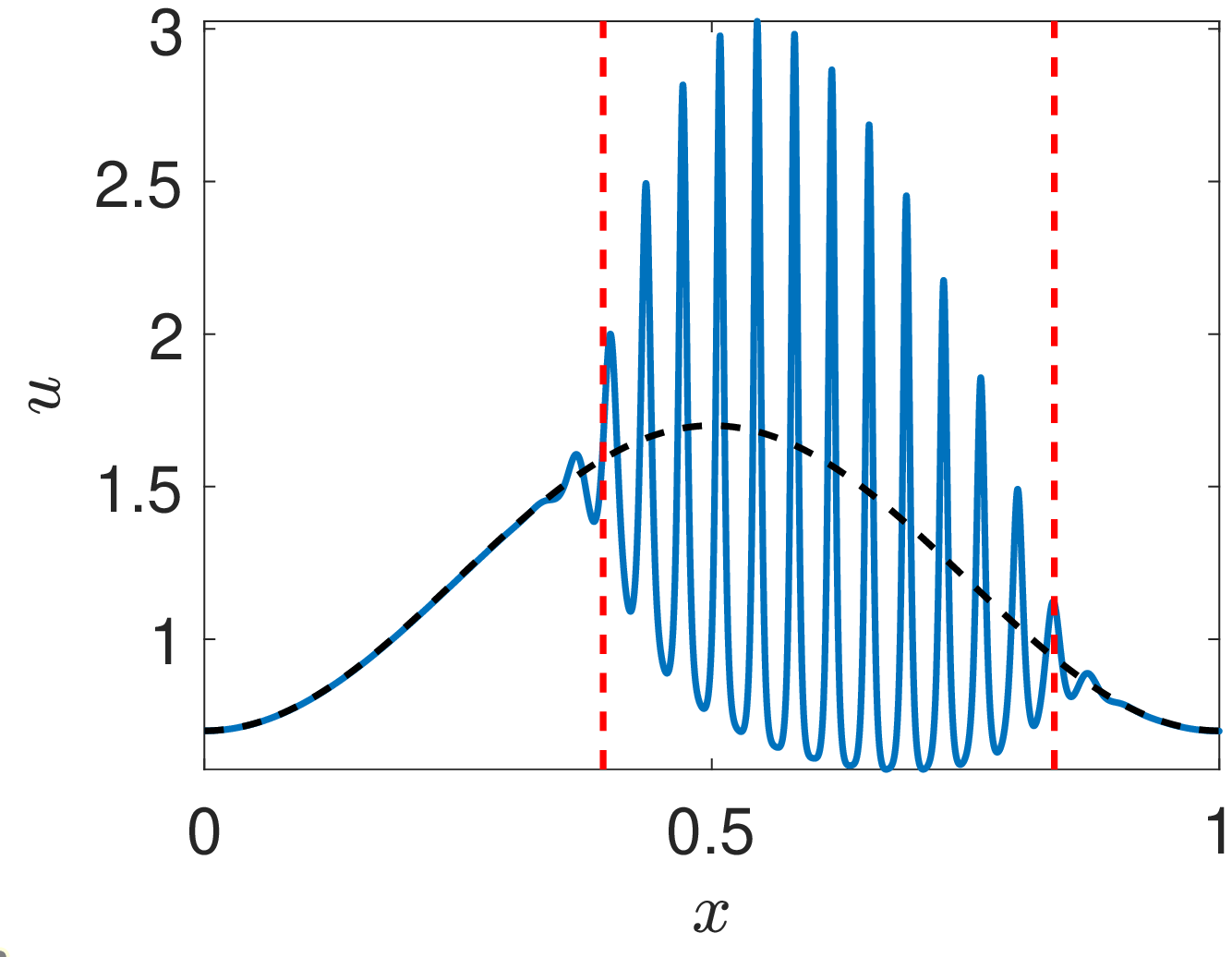}}
    
    \subfloat[$\ep=0.003$]{\includegraphics[width=0.49\textwidth]{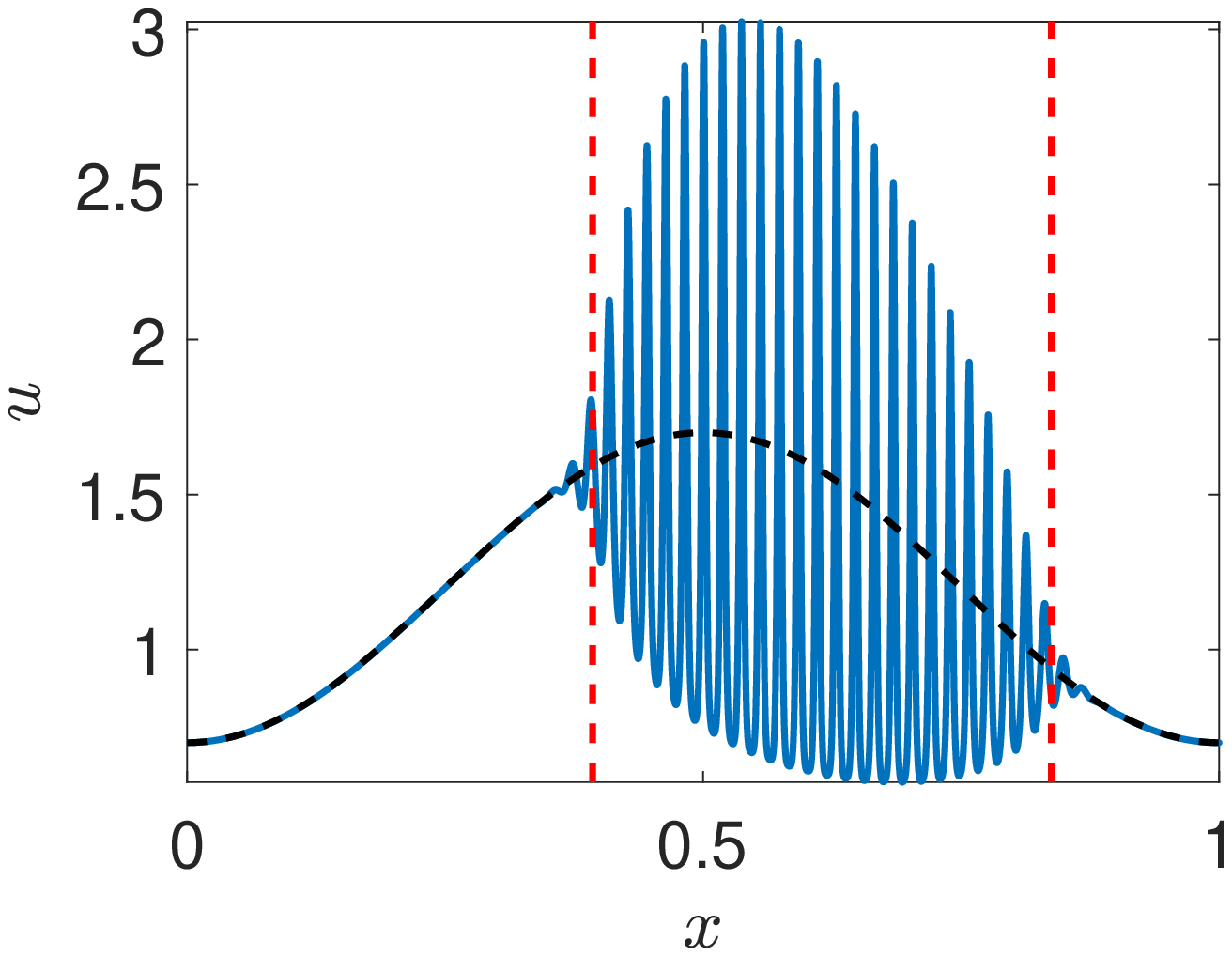}}\hspace{0.2cm}\subfloat[$\ep=0.001$]{\includegraphics[width=0.49\textwidth]{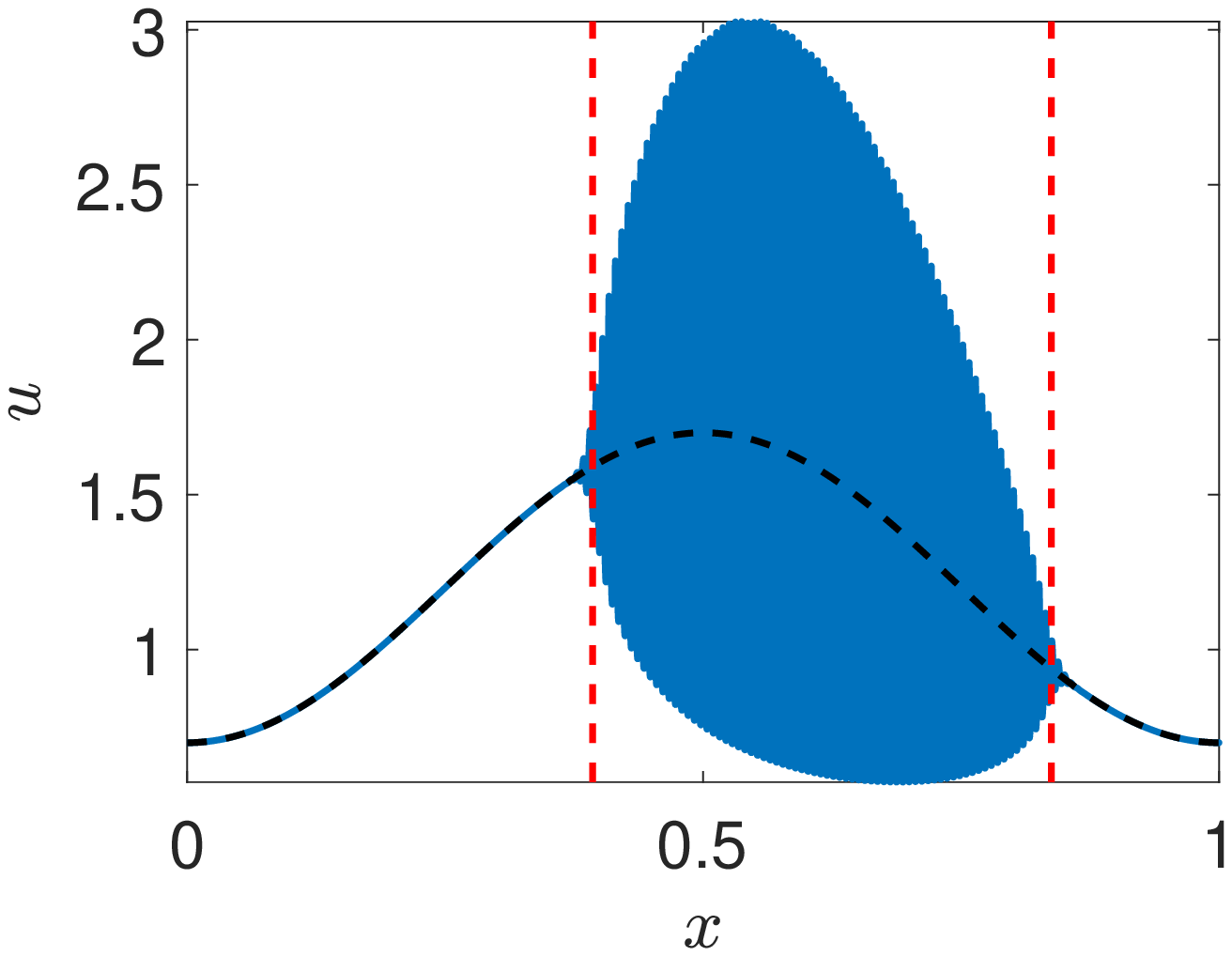}}

    \caption{Plots of $u$ in blue curves and $u^*$ in dashed black curves from solutions {of the 
     Keller-Segel model,  \cref{KS_eqns},} for various values of $\ep$ at $T=50,000$. The red  {vertical }  lines show the boundary of $\mathcal{T}_0$ computed from the conditions in \cref{het_prop}. The parameters are taken as $K(x) = 1.2-0.5\cos(2\pi x)$, $h(x) = (1-0.5\cos(\pi x))$, $D_{11}=D_{22}=1$, and $\chi(x) = 3.05-0.1x$. }
    \label{KSStat}
\end{figure}

\begin{figure}
    \centering
    
    \subfloat[$\ep=0.02$, $h = 0$, \\$\chi =3.05-0.1x $]{\includegraphics[width=0.49\textwidth]{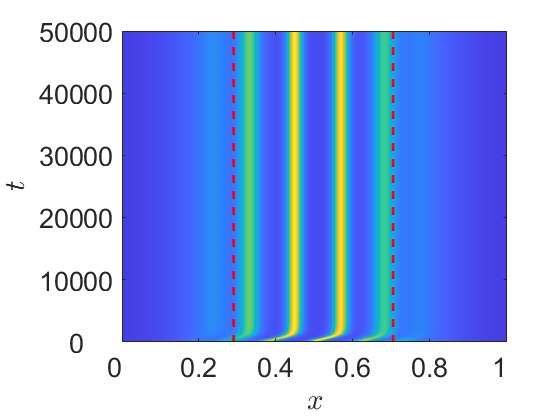}}\hspace{0.2cm}\subfloat[$\ep=0.02$, $\chi = 0$ \\$h = (1-0.5\cos(\pi x))$ ]{\includegraphics[width=0.49\textwidth]{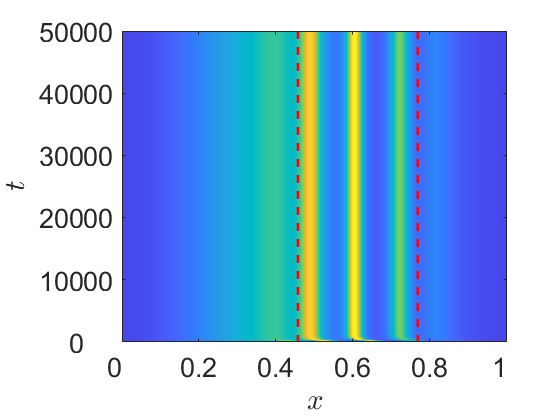}}
    
    \subfloat[$\ep=0.01$, $h = 0$, \\$\chi =3.05-0.1x$]{\includegraphics[width=0.49\textwidth]{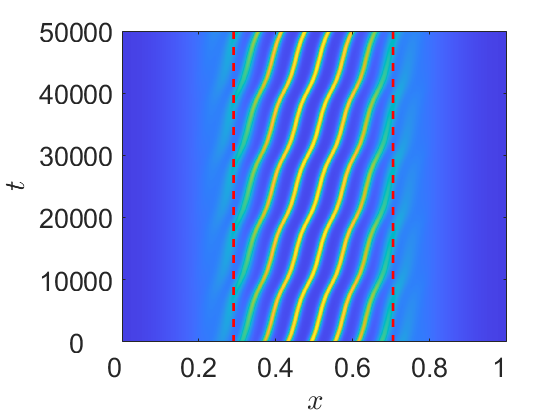}}\hspace{0.2cm}\subfloat[$\ep=0.01$, $\chi = 0$ \\$h = (1-0.5\cos(\pi x))$]{\includegraphics[width=0.49\textwidth]{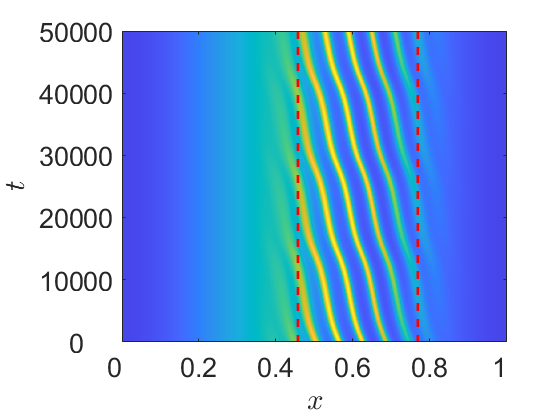}}

    \subfloat[$\ep=0.005$, $h = 0$, \\$\chi =3.05-0.1x$]{\includegraphics[width=0.49\textwidth]{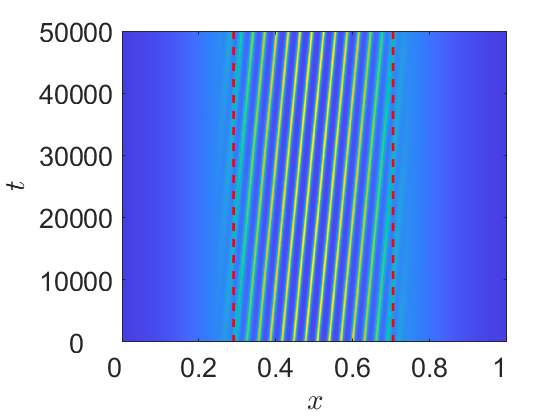}}\hspace{0.2cm}\subfloat[$\ep=0.005$, $\chi = 0$ \\$h = (1-0.5\cos(\pi x))$]{\includegraphics[width=0.49\textwidth]{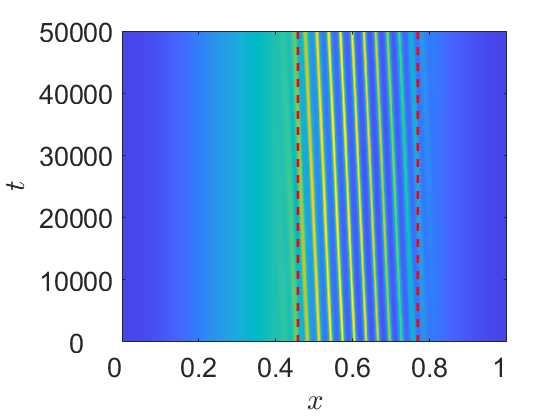}}
    \caption{Plots of $u$ in blue curves and $u^*$ in dashed black curves from solutions {of the 
     Keller-Segel model,  \cref{KS_eqns},}  for various values of $\ep$ at $T=50,000$. The red  {vertical }  lines show the boundary of $\mathcal{T}_0$ computed from the conditions in \cref{het_prop}. The parameters are taken as $K(x) = 1.2-0.5\cos(2\pi x)$ and  $D_{11}=D_{22}=1$. }
    \label{KSMov}
\end{figure}

We next consider simulations of the Keller-Segel chemotaxis system \cref{KS_eqns}. We show example steady state behaviors in \cref{KSStat} for heterogeneous $h(x)$, $K(x)$, and $\chi(x)$. While there is a size and wavelength modulation, and the patterned region is asymmetrically shaped, such solutions are steady in time. In contrast, if we let either $\chi(x)=0$ or $h(x)=0$, we observe the more complicated spatiotemporal behavior in \cref{KSMov} for sufficiently small $\ep$. The direction of spike movement in either case is opposite, which helps explain why their combination leads to steady spikes in \cref{KSStat}. Interestingly, while the spike movement decreases with decreasing $\ep$, for $\ep=0.02$, no spike movement is observed in panels (A) and (B) of \cref{KSMov}. This is consistent with observations of heterogeneous reaction-diffusion systems in \cite{krause2018heterogeneity}, as such spike oscillations arise due to a global bifurcation structure involving spike creation and annihilation.

    \begin{figure}
    \centering
    
    \subfloat[$\ep=0.006$ ]{\includegraphics[width=0.49\textwidth]{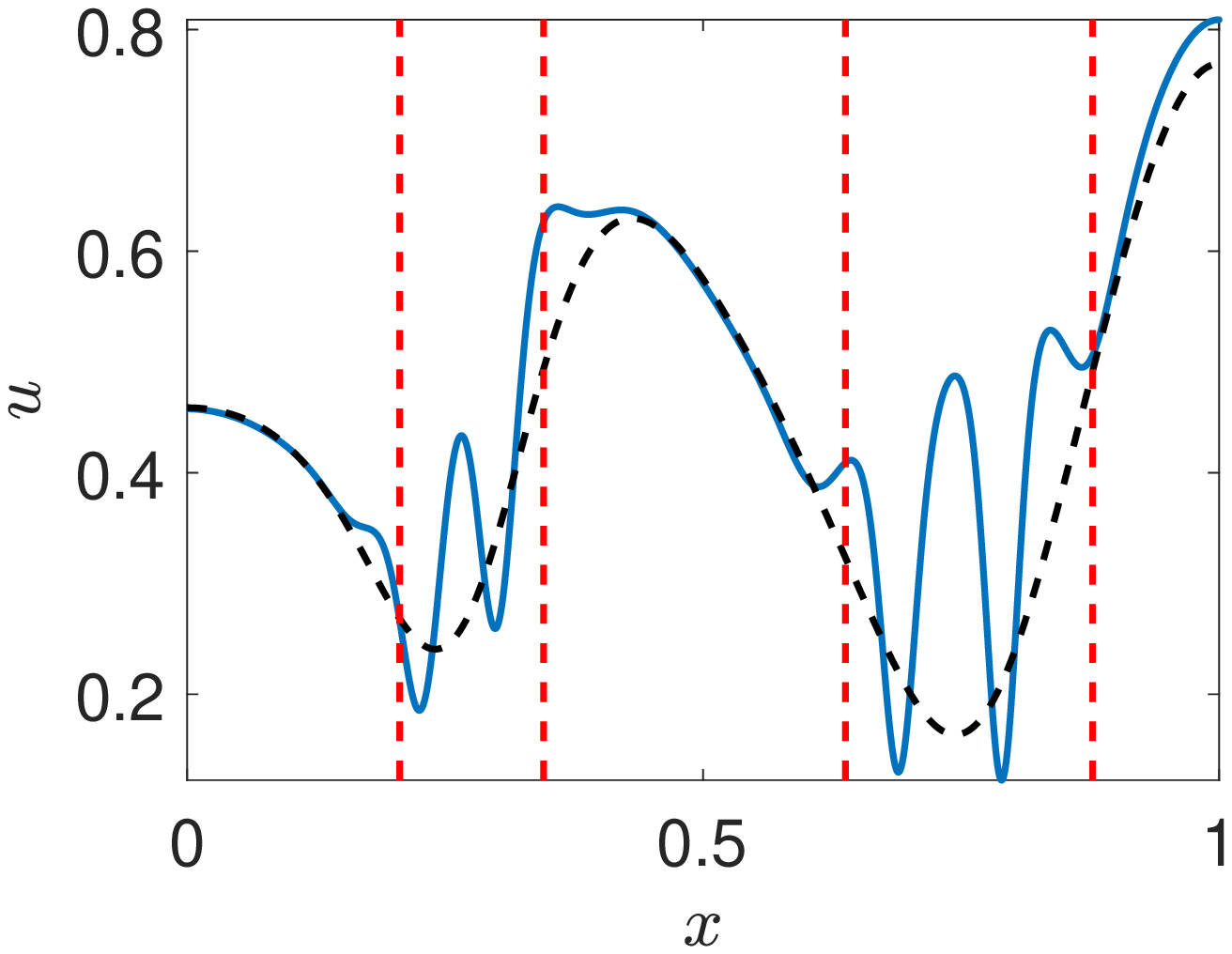}}\hspace{0.2cm}\subfloat[$\ep=0.006$]{\includegraphics[width=0.49\textwidth]{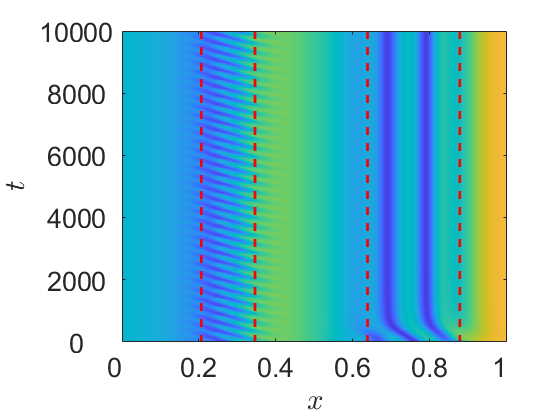}}
    
    \subfloat[$\ep=0.004$ ]{\includegraphics[width=0.49\textwidth]{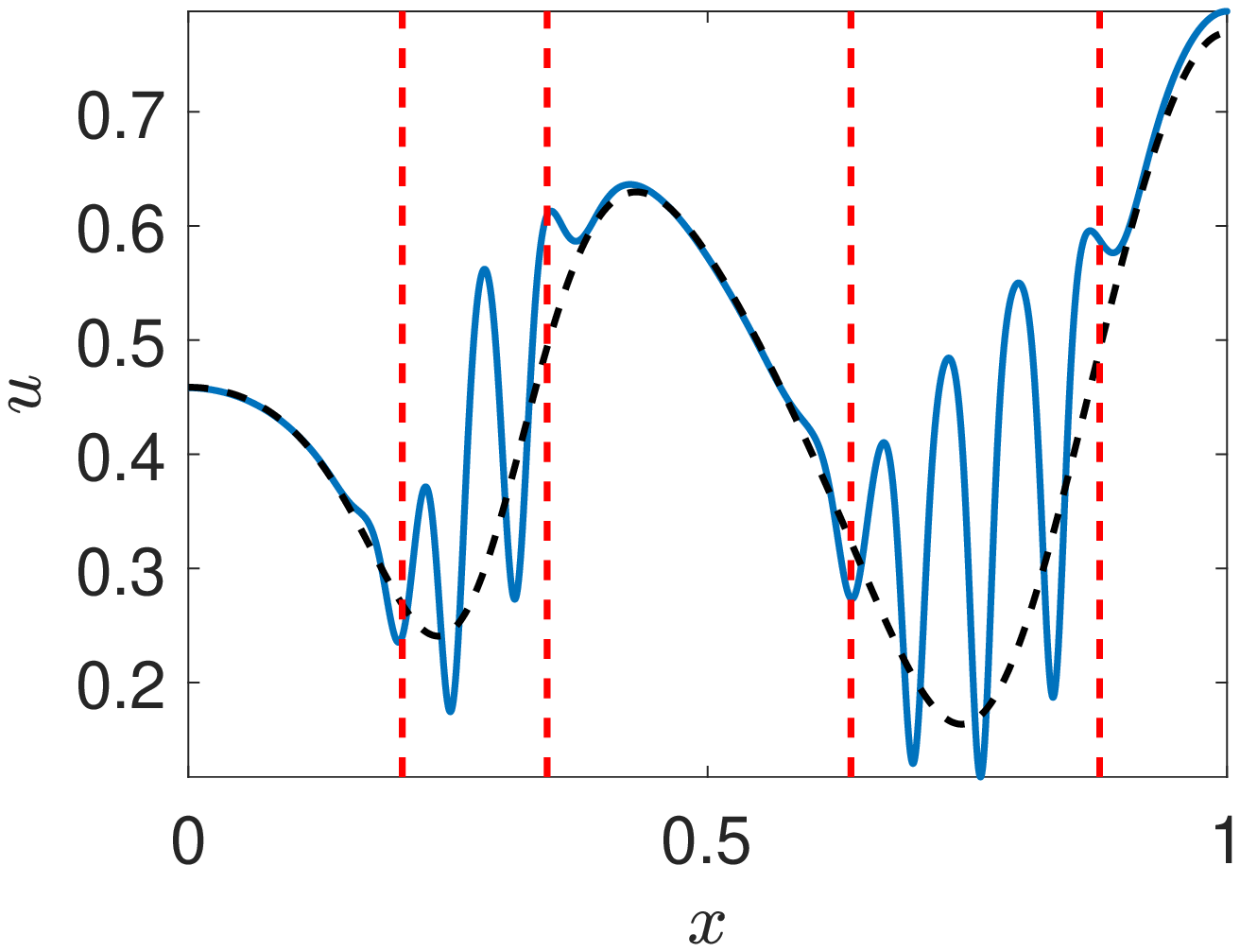}}\hspace{0.2cm}\subfloat[$\ep=0.004$ ]{\includegraphics[width=0.49\textwidth]{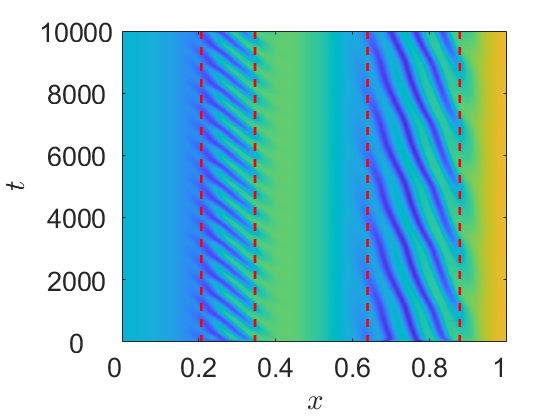}}
    
    \subfloat[$\ep=0.001$ ]{\includegraphics[width=0.49\textwidth]{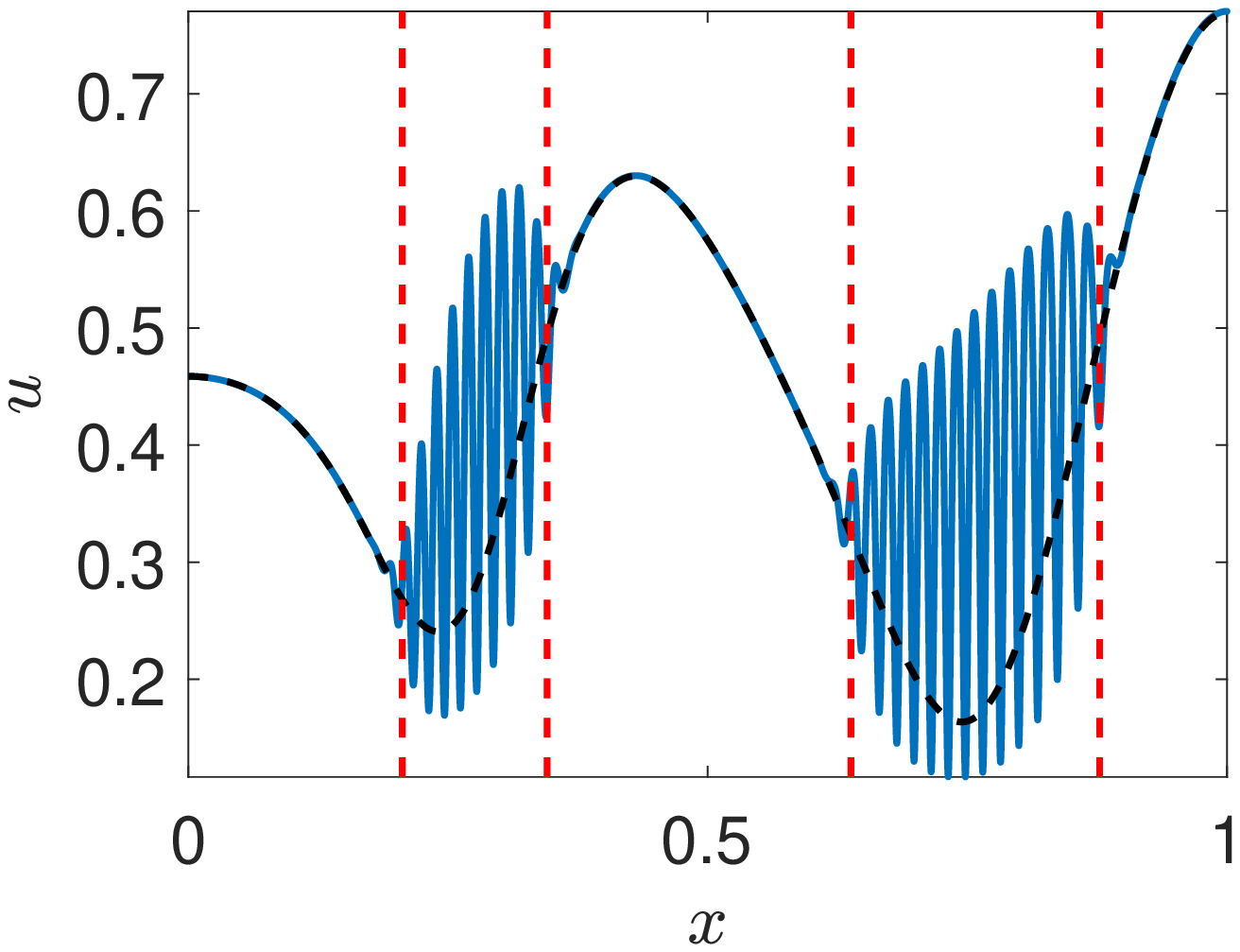}}\hspace{0.2cm}\subfloat[$\ep=0.001$ ]{\includegraphics[width=0.49\textwidth]{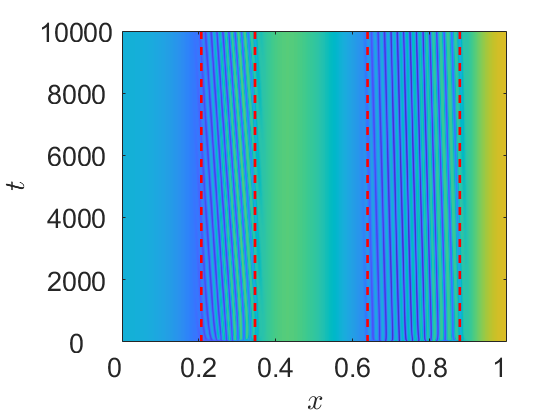}}
    \caption{Plots of $u$ in blue curves and $u^*$ in dashed black curves from solutions of {the SKT model, \cref{SKT_eqns},} for various values of $\ep$ at $T=50,000$ in (A), (B), (E), and kymographs of $u$ in (B), (D), (F). The red  {vertical }  lines show the boundary of $\mathcal{T}_0$ computed from the conditions in \cref{het_prop}. The parameters are taken as $r_1(x)=1$, $r_2(x)=2$, $a_1(x)=0.9+0.2\cos(3\pi x)$, $a_2(x) = 0.9+0.2\cos(4\pi x)$, $b_1(x)=0.6$, $b_s(x)=0.2$, $d_1(x)=d_2(x)=1$, $d_{21}(x)=200x$, $d_{11}(x)=d_{12}(x)=d_{22}(x)=0$.}
    \label{SKTFig}
\end{figure}

    We now demonstrate localized pattern formation in the SKT model given by \cref{SKT_eqns}. We consider heterogeneous kinetic and diffusion parameters in \cref{SKTFig}, observing again that while localization occurs approximately within the bounds predicted by \cref{het_prop}, there is spike movement for sufficiently small $\ep$. We note that the spike speed decreases for decreasing values of $\ep$, and spike wavelength, amplitude, and speed are spatially-dependent. We also observe in panel (B) that for sufficiently large $\ep$, stationary patterns are observed in part of the domain, while spatiotemporal movement is still observed in another region. For $\ep \geq 0.007$, there is no longer any spike movement and all of the solutions we found tended to stationary spatial profiles.

\section{Discussion}\label{Discuss_Sect}
We have generalized the results found in \cite{krause_WKB} to nonlinear reaction-cross-diffusion systems. Our main result, presented in \cref{het_prop}, can be interpreted as giving the locations where we expect pattern formation which is emergent from a Turing-type instability, as opposed to spatial structure arising from background heterogeneity itself. Such a distinction is a key result of our {linear} theory. We numerically validated this theory using a wide variety of cross-diffusion models, finding an excellent agreement with the analytical predictions for sufficiently small $\ep$. We also showed that numerical instabilities arising from spike instabilities and creation will also remain within the spatial regions predicted by our instability analysis, though our theory cannot differentiate between steady state pattern formation and such spatiotemporal oscillations.

There are many avenues for future work, particularly where we think the main approach presented may be especially tractable. In principle our theory of pattern localization should extend to systems where all spatial derivatives of order $n$ are scaled with $\ep^n$, such as in fourth-order pattern forming models like Swift-Hohenberg or Cahn-Hilliard systems. Such a setting may actually be an easier context in which to pursue weakly nonlinear analyses to study the saturation of pattern amplitude as a function of the growth rate $\lambda$, and potentially to explore questions of spatiotemporal oscillations. Numerical continuation could help connect the theory developed here to other kinds of localized patterns, such as localized patterns found via nonlinear mechanisms such as homoclinic snaking \cite{uecker2014numerical, al2021localized}. An interesting question would be how these different kinds of localized structures interact.

Extending these results to more than two species, as has been well-studied in the spatially homogeneous setting \cite{satnoianu2000turing}, is in principle straightforward, though the calculations may become overly cumbersome. A more difficult challenge, and one for which we have no easy way to extend the theory, is to prove analogous results for higher {spatial} dimensions. Numerically, we have explored such extensions and found excellent agreement with the expected localization \cite{woolley2021bespoke}, but the WKB approach becomes substantially more technical to use in higher dimensions. We anticipate that there is an alternative way to prove something analogous to \cref{het_prop} which can be extended to higher dimensions. There are also aspects of our theory which deserve a more careful rigorous development, which others have begun \cite{kovavc2022liouville}. 

While the details of our approach are somewhat technical, the overall results are intuitive, and in some sense can be viewed as justifying the `obvious' localization one might anticipate in the regime of small $\ep$. Nevertheless, our framework opens up a variety of ways of thinking of localization in pattern forming systems, which we think has ample use in developmental biological and ecological settings. The model \cref{orig_eqn} inherently contains two scale-separation assumptions. The explicit assumption in $\ep$ is that the diffusive scaling is smaller than that of the reaction (and that gradients in heterogeneous reactions are sufficiently slow). A second inherent assumption, however, is a timescale separation between the proposed pattern-forming model and whatever led to the pre-patterned state giving rise to the explicit spatial heterogeneities. Determining when these assumptions are valid, or how to conceptualize these systems when these assumptions are relaxed, is another key area of future work.


\bibliographystyle{AIMS}
\bibliography{refs}

\medskip
Received xxxx 20xx; revised xxxx 20xx.
\medskip

\end{document}